\documentclass[prd,aps,amsfonts,showpacs,longbibliography,notitlepage,nofootinbib]{revtex4-1}
\usepackage[dvipsnames]{xcolor}
\usepackage{bm}
\usepackage{bbm}

\usepackage{amsmath,amssymb,amsfonts,mathtools,dsfont,relsize}
\usepackage{graphicx}

\usepackage[colorlinks=true, urlcolor=violet, linkcolor=blue, citecolor=red, hyperindex=true, linktocpage=true]{hyperref}

\usepackage{amsthm}
\usepackage{qtree}
\newtheorem{thm}{Theorem}
\newtheorem{cor}[thm]{Corollary}
\newtheorem{lem}[thm]{Lemma}
\newtheorem{prop}[thm]{Proposition}

\newtheorem{defn}[thm]{Definition}

\usepackage{soul} 

\makeatletter
\newtheorem*{rep@theorem}{\rep@title}
\newcommand{\newreptheorem}[2]{%
\newenvironment{rep#1}[1]{%
 \def\rep@title{#2 \ref{##1}}%
 \begin{rep@theorem}}%
 {\end{rep@theorem}}}
\makeatother
\newreptheorem{thm}{Theorem}
\newreptheorem{cor}{Corollary}
\newreptheorem{defn}{Definition}

\usepackage{enumerate}
\usepackage{stmaryrd} 

\renewcommand{\thesection}{\arabic{section}}
\renewcommand{\thesubsection}{\thesection.\arabic{subsection}}
\renewcommand{\thesubsubsection}{\thesubsection.\arabic{subsubsection}}
\makeatletter
\renewcommand{\p@subsection}{}
\renewcommand{\p@subsubsection}{}
\makeatother

\DeclarePairedDelimiter{\floor}{\lfloor}{\rfloor}

\DeclarePairedDelimiter{\ket}{\lvert}{\rangle}
\DeclarePairedDelimiter{\bra}{\langle}{\rvert}

\usepackage{qcircuit}

\newcommand{\lr}[1]{\left( #1\right)}
\newcommand{\mlr}[1]{\left[ #1\right]}
\newcommand{\glr}[1]{\left\{ #1\right\}}
\newcommand{\alr}[1]{\left\langle #1\right\rangle}
\newcommand{\norm}[1]{\left\lVert#1\right\rVert}
\newcommand{\abs}[1]{\left\lvert#1\right\rvert}
\newcommand{\rarrow}{\quad \Rightarrow \quad}
\newcommand{\ee}{\mathrm{e}}
\newcommand{\dd}{\mathrm{d}}

\newcommand{\OO}{\mathcal{O}}

\newcommand{\cA}{\mathcal{A}}

\newcommand{\PP}{\mathbb{P}}

\makeatletter 
    
\renewcommand\onecolumngrid{
\do@columngrid{one}{\@ne}%
\def\set@footnotewidth{\onecolumngrid}
\def\footnoterule{\kern-6pt\hrule width 1.5in\kern6pt}%
}

\renewcommand\twocolumngrid{
        \def\footnoterule{
        \dimen@\skip\footins\divide\dimen@\thr@@
        \kern-\dimen@\hrule width.5in\kern\dimen@}
        \do@columngrid{mlt}{\tw@}
}%

\makeatother

\newcommand{\where}{\quad \mathrm{where}\quad}

\newcommand{\diam}{\mathrm{diam}}

\newcommand{\bs}{\mathbf{s}}
\newcommand{\chec}{C}
\newcommand{\tv}{\widetilde{v}}
\newcommand{\td}{\widetilde{\mathbbm{d}}}

\newcommand{\revise}[1]{#1}
\newcommand{\comment}[1]{}

\begin{document}
\title{Low-density parity-check codes as stable phases of quantum matter}

\author{Chao Yin}
\email{chao.yin@colorado.edu}
\affiliation{Department of Physics and Center for Theory of Quantum Matter, University of Colorado, Boulder, CO 80309, USA}

\author{Andrew Lucas}
\email{andrew.j.lucas@colorado.edu}
\affiliation{Department of Physics and Center for Theory of Quantum Matter, University of Colorado, Boulder, CO 80309, USA}

\date{\today}

\begin{abstract}

Phases of matter with robust ground-state degeneracy, such as the quantum toric code, are known to be capable of robust quantum information storage. Here, we address the converse question: given a quantum error correcting code, when does it define a stable gapped quantum phase of matter, whose ground state degeneracy is robust against perturbations in the thermodynamic limit?  We prove that a low-density parity-check (LDPC) code defines such a phase, robust against all few-body perturbations, if its code distance grows at least logarithmically in the number of degrees of freedom, and it exhibits ``check soundness".  Many constant-rate quantum LDPC expander codes have such properties, and define stable phases of matter with a constant zero-temperature entropy density, violating the third law of thermodynamics.  Our results also show that quantum toric code phases are robust to spatially nonlocal few-body perturbations.  Similarly,  phases of matter defined by classical codes are stable against symmetric perturbations.  In the classical setting, we present improved locality bounds on the quasiadiabatic evolution operator between two nearby states in the same code phase.
\end{abstract}

\maketitle

\section{Introduction}
The past two decades have seen a revolution in our understanding of quantum phases of matter \cite{QI_meet_QM,phase_localU10,phase_LRB06,topo_Hastings,phase_adaptive21,phase_adaptive_hierarchy23,phase_nonlocal24}.    

First and foremost, we now have a useful definition of a gapped phase of quantum matter: two states $|\psi_1\rangle$ and $|\psi_2\rangle$ are in the same phase if there is some unitary $U$, obtained by finite-time evolution with a (quasi)local Hamiltonian, such that $|\psi_1\rangle = U|\psi_2\rangle$.  
Notice that this is very different from the textbook thermodynamic definition of a phase of matter, characterized by passing from one system to another without encountering any thermodynamic singularities in, e.g., a free energy density.  

Secondly, we have discovered new kinds of quantum phases that are completely robust to all local perturbations. A simple example is the quantum toric code \cite{toric_code03}, a simple model exhibiting local topological order \cite{topo_Hastings,michalakis2013stability}.  Local topological order implies that all ground states are indistinguishable to any local operators, in stark contrast to the Landau paradigm for phases of matter, wherein \emph{local} order parameters (such as the magnetization, for a ferromagnetic phase) can be used to diagnose a symmetry breaking pattern. 

The quantum toric code both exhibits local topological order, and also highlights why the thermodynamic definition of a phase falls short: it is thermodynamically indistinguishable from a trivial paramagnet \cite{Brown_2016, Weinstein_2019} yet corresponds to a topologically ordered \emph{zero temperature} phase of matter. It also represents the simplest example of both a topological quantum phase, \emph{and} a quantum error-correcting code robust against finitely many errors, in the thermodynamic limit.  Indeed, the ability of the toric code to correct against errors is a simple consequence of local topological order, as no local perturbation can transform one ground state to another, i.e. induce a logical error on quantum information protected in the ground state(s).  Many experiments \cite{TC_exp22,TC_IBM23,TC_google23,TC_google24,Ryd_Lukin24,TC4d_Quantinuum24} have aimed to realize scalable error correction of a logical qubit using the toric (or related surface) code.  

Since the toric code only stores a handful of logical qubits in the thermodynamic limit, however, many researchers are also looking beyond the toric code for a more general family of low-density parity check (LDPC) codes.  LDPC codes, like the toric code, can be defined in terms of quantum Hamiltonians with few-body interactions.  Like the toric code, the ground states of these Hamiltonians represent the entangled logical codewords; unlike the toric code, the interactions are not local in any finite spatial dimension.  Instead, LDPC codes are ``$q$-local": the Hamiltonian is a sum of products of $q$ or fewer Pauli matrices acting on $n$ qubits, but there are no constraints on which $q$ qubits can interact in any given term.

There has been a flurry of recent activity on discovering interesting quantum LDPC codes \cite{HGP,TP_codes,BP_codes,LP_codes,Panteleev_2022,qTanner_codes,dinur2022good}, understanding the physics of these new quantum codes \cite{Hong:2024vlr,Rakovszky:2023fng,Rakovszky:2024iks}, and adapting them in real experiments \cite{LRESC_Yifan,Hong:2024fsg,HGP_Rydberg24,HGP_gate_Ryd24}.  At low temperature, some of these codes represent highly entangled phases of matter where all low-energy states lie in a nontrivial phase of matter \cite{Anshu_2023} (verifying the ``NLTS conjecture" of \cite{freedman2013}).  These low energy states are capable of quantum self-correction \cite{Hong:2024vlr}.  It has further been conjectured that the good LDPC codes, defined on expander graphs, form well-defined phases of matter \cite{stability_maryland,stability_Khemani24}; a proof for non-expander codes is in \cite{stability_maryland}.  These would represent unusual phases, with a completely robust violation of the third law of thermodynamics, with an extensive ground state degeneracy in the thermodynamic limit, as emphasized in \cite{stability_maryland}.  However, because LDPC codes are only $q$-local -- not spatially local -- the mathematical tools which have, until now, been able to prove the stability of topological order, cannot be applied.

Here, we prove that many quantum LDPC expander codes indeed form stable phases of matter to $q$-local perturbations.  In particular, we prove that (\emph{i}) the gap between the ground states, which encode logical information in the code, and excited states, does not close for arbitrary small perturbations; (\emph{ii}) the splitting of the energy levels between these code states is stretched-exponentially small in $n$ (for many codes of interest); (\emph{iii}) a quasilocal unitary can rotate from the unperturbed to the perturbed ground state sector efficiently.  These ideas are summarized in Fig.~\ref{fig:spec}.  Our proof was inspired by the techniques we developed in a recent theory of quantum metastable states and slow false vacuum decay \cite{our_metastable}. We identify a sufficient condition on a quantum code for stability, coined ``check soundness", which is obeyed by finite-dimensional toric codes, hypergraph product expander codes \cite{HGP}, and the ``good" constant-rate quantum LDPC codes \cite{Panteleev_2022,qTanner_codes,dinur2022good}.  The latter two constant-rate quantum codes represent a truly ``infinite-dimensional" phase of matter without a third law of thermodynamics, whose stability to perturbations is established here.  

Even in finite dimensional models, our methods lead to a number of new results.    We establish the stability of quantum toric code phases to $q$-local perturbations, which is consistent with the nontrivial circuit depth lower bounds \cite{phase_LRB06,code_circuit_lowerbound18,phase_nonlocal24} required to prepare such topological states. As in \cite{bravyi2011short}, it is also simple to extend our results to study the stability of phases of matter defined by classical codes, in the presence of symmetric perturbations.  An important application of this result concerns the construction of quasiadiabatic evolution operators that connect two different states in the same phase.  By avoiding spectral filtering \cite{quasiadiabatic05,bachmann2012automorphic}, we are able to find very strong bounds on the locality of these operators, one application of which will be detailed elsewhere \cite{volumeLaw_disorder24}.

The remainder of the main text gives a high-level and non-rigorous introduction to quantum LDPC codes, our main theorem, and some implications of the result.  Appendices contain technical definitions along with our main theorem's statement and proof.

\begin{figure}
    \centering
    \includegraphics[width=0.5\linewidth]{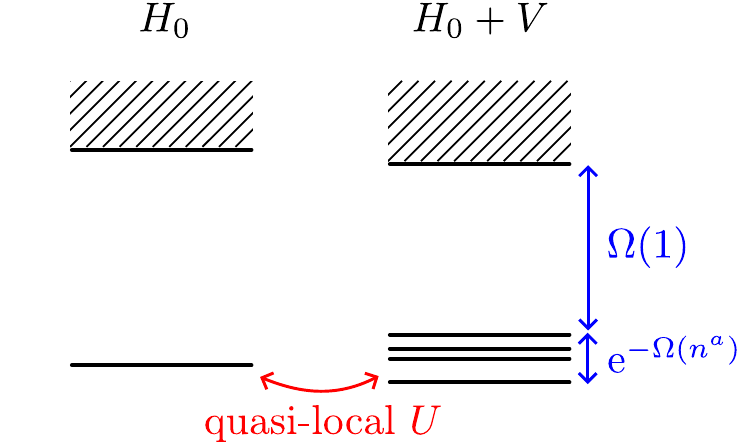}
    \caption{We prove that many quantum LDPC codes are stable against perturbations, in the sense that (\emph{i}) the gap between the encoded ground states and excited states remains $\Omega(1)$, (\emph{ii}) the energy splitting of the ground states is exponentially small, (\emph{iii}) the perturbed ground states are connected to the unperturbed ones by a sufficiently local unitary. }
    \label{fig:spec}
\end{figure}

\section{Error correcting codes}
We first provide a high-level introduction to (quantum) error-correcting codes, from a physicist's point of view.

The simplest kind of error-correcting code -- the repetition code -- can be understood quite readily as the classic Ising model from statistical physics.  On a $d$-dimensional lattice, we have Hamiltonian \begin{equation}\label{eq:IsingH0}
    H_0 = -\sum_{i\sim j} Z_i Z_j
\end{equation}
where $i\sim j$ denotes that $i,j$ are neighbors on the lattice, and $Z_i \in \lbrace \pm 1\rbrace$ are binary variables.  $H_0$ has a doubly-degenerate ground state, corresponding to either all $Z_i = +1$ or all $Z_i = -1$.  These two ground states encode a logical classical bit.  In $d>1$ dimensions, even when coupling $H_0$ to a thermal bath, the ferromagnetic low-temperature phase of the Ising model ensures that the time it takes to move from one ground state to another is exponentially large in the system size \cite{Thomas1989}; this is the basis of magnetic memory.   We say that the Ising model stores $k=1$ logical bit in $n$ physical bits.   To get from one logical codeword to another, we must flip every single bit, so the code distance $d=n$:  if we have $\frac{1}{2}(d-1)$ errors in the system, then an optimal decoder could correctly return us to our original state.  In fact, thinking of $H_0$ as a quantum Hamiltonian acting on a Hilbert space of $n$ qubits, with Pauli matrices \begin{equation}
    Z = \left(\begin{array}{cc} 1 &\ 0 \\ 0 &\ -1 \end{array}\right), \;\;\;\;  X = \left(\begin{array}{cc} 0 &\ 1  \\ 1 &\  0  \end{array}\right),
\end{equation}
we can interpret code distance $d$ as the size of the smallest logical operator $X_1\cdots X_n$ that converts one code state $|00\cdots 0\rangle$ to the other: $|11\cdots 1\rangle$.

The simplest quantum error correcting code whose distance $d$ scales with $n$ is the two-dimensional toric (or surface) code \cite{toric_code03}.   Consider a two-dimensional square lattice, with qubits placed on the edges $e$ of the graph.  The codewords are the ground states of the Hamiltonian \begin{equation}
    H_0 = -\sum_{\text{vertex } v} \prod_{e\sim v} X_e -\sum_{\text{face } f} \prod_{e\sim f} Z_e = -\sum_{\text{vertex } v} A_v -\sum_{\text{face } f} B_f, \label{eq:toriccodeH}
\end{equation}
where $e\sim v$ means that edge $e$ is adjacent to vertex $v$, and $e\sim f$ means that edge $e$ lies on the boundary of face (plaquette) $f$.  We have also introduced the $A_v$ and $B_f$ notation, as is common in the literature.  A simple calculation shows that \begin{equation}
    [A_v, B_f] = 0  \label{eq:toriccommutator}
\end{equation}
for all $v$ and $f$, because these operators always share an even number of edges in common, and thus an even number of Pauli $X$ and $Z$ on the same edges (whose product therefore commutes).   \eqref{eq:toriccommutator} implies that the ground state of $H_0$ is the simultaneous $+1$ eigenspace of the operators $A_v$ and $B_f$, which are called the parity-checks of the code.  They are analogous to $Z_iZ_j$ in the Ising model, but due to the presence of both Pauli $X$ and $Z$, imply the existence of large entanglement in the ground states of $H_0$. The ground state subspace is also the $+1$ eigenspace of any stabilizer, i.e. a Pauli string that is a product of the checks $A_v,B_f$. Logical operators for the toric code become products of Pauli $X$ and $Z$ that commute with all stabiliers, and correspond to string operators across the system.

In a pioneering set of works, \cite{topo_Hastings,bravyi2011short} showed that for any $V$ which is a sum of small terms (but which may be thermodynamically extensive), the toric code Hamiltonian $H_0$ is robust against perturbations: \begin{equation}
    H = H_0 + V
\end{equation}
lies in the same phase as $H_0$, with the same ground state degeneracy (up to exponentially small corrections in system size), and with a quasilocal unitary that can rotate from the ground states of $H_0$ to the ground states of $H$.   To find such an absolutely stable phase of matter is quite non-trivial at first glance: after all, the ground state degeneracy of the Ising $H_0$ is \emph{not} stable in this way: taking \begin{equation}
    V = \frac{1}{n}\sum_i Z_i \label{eq:falsevacuumperturbation}
\end{equation}
already completely breaks the ground state degeneracy.   The absence of stability of the Ising model to these longitudinal fields is quite well-known, and related to the physics of the false vacuum \cite{our_metastable}.  

The language of error correction helps to explain the discrepancy.  The toric code is a \emph{quantum code}:  acting on a codeword of the toric code with $\ll d$ errors, either $X$-type or $Z$-type, can be detected and corrected by identifying the locations of the flipped parity-checks $A_v$ and $B_f$.   Viewing the classical Ising model as a quantum code, we see that it has a large distance $d_X=n$ for $X$ errors, but a tiny distance $d_Z=1$ for $Z$-errors, as \begin{equation}
    Z_1 \frac{|00\cdots 0\rangle + |11\cdots 1\rangle }{\sqrt{2}} = \frac{|00\cdots 0\rangle - |11\cdots 1\rangle }{\sqrt{2}}. \label{eq:isingxonly}
\end{equation}
If, therefore, we study the robustness of the Ising model $H_0$ to peturbations $V$ that are products of $X$ operators, we might expect a similar stability to the toric code, and this has indeed been shown rigorously \cite{bravyi2011short,longrange_stab23}.  The intuitive result is much older, and dates back to the classic solution of the transverse field Ising model, which undergoes a quantum phase transition out of the ferromagnetic phase only for a finite perturbation strength in $V$ \cite{sachdev_book}.

It thus appears that whenever we have a quantum Hamiltonian $H_0$, which can be interpreted as a quantum error correcting code, we might expect a stable phase of matter.  Along these lines, we might ask what happens when we study more exotic codes than those which correspond to geometrically local Hamiltonians $H_0$.  The most famous example of such a code is an LDPC code.  The details of how to construct $H_0$ for an LDPC code will not be essential to understanding our main results, but at a high-level, one writes \begin{equation}
    H_0 = -\sum_{\text{parity-check } C_X} \prod_{i\in C_X} X_i -\sum_{\text{parity-check } C_Z} \prod_{i\in C_Z} Z_i,
\end{equation}
subject to (\emph{i}) the requirement that each $C_X$ and $C_Z$ commute and (\emph{ii}) the requirement that each $C_{X,Z}$ is a product of O(1) Paulis, and that each Pauli $X_i$ or $Z_i$ shows up in an O(1) number of parity-checks.\footnote{ Here the big-O notation $y=\mathrm{O}(x)$ means $y\le c x$ for a constant $c$ independent of system size. We will also use $y=\Omega(x)$ for $y\ge Cx$ with a constant $C$, and $y=\Theta(x)$ for $cx\le y\le Cx$.} Remarkably, it is known how to construct such codes which are constant-rate: in the thermodynamic limit $n\rightarrow \infty$, one finds $k=\Theta(n)$, while $d$ also grows with $n$, ideally linearly.  From an error-correcting perspective, such codes only use a constant number of redundant physical bits to encode information, yet in an encoding which is extremely robust to large errors.  It is known \cite{BPT_bound10} that such constant-rate codes (with growing distance) can only exist in the limit of infinitely-many spatial dimensions.  Mathematically, one often refers to such a code as defined on an expander graph, where two qubits share an edge if they share a common parity-check, and the resulting graph has the property that \begin{equation}
    \text{number of qubits within distance $r$ of qubit $i$} \gtrsim \exp[\alpha r], \;\;\; \text{ if } r\lesssim \log n. \label{eq:numberexpandergraph}
\end{equation}
for some positive $\alpha>0$.

The purpose of this paper is to make this intuition precise, and to prove that a quantum LDPC Hamiltonian $H_0$ indeed forms a stable phase of matter.   A classical LDPC code (with only $Z$-type checks) will similarly form a robust phase of matter to $X$-type perturbations.

\section{Check soundness: a sufficient condition for stability}
To prove the stability of an LDPC code, we must first ask whether there is any way to make a genuine code unstable. While the ``reasonable" error correcting toric code Hamiltonian in \eqref{eq:toriccodeH} indeed represents a stable phase of matter, a ``stupid" toric code can fail to have this property \cite{topo_Hastings}.  Consider modifying \eqref{eq:toriccodeH} to 
\begin{equation}\label{eq:Ising_toric}
    H_0 = -\sum_{v} A_v - B_{f_0} - \sum_{f\sim f^\prime} B_f B_{f'},
\end{equation}
where $A_v,B_f$ are the vertex and plaquette checks introduced before, $f_0$ is one ``special" face anywhere in the lattice, and faces $f\sim f^\prime$ are adjacent if and only if they share a single edge. In the modified Hamiltonian \eqref{eq:Ising_toric}, only the single plaquette term at $p_0$ remains to be a check, and the other $Z$-checks are Ising interactions among all nearest-neighbor plaquettes $f\sim f^\prime$. Although the modified $H_0$ has the same error-correcting ground states as the standard toric code, its gap is unstable to the same kind ``longitudinal field" perturbation as \eqref{eq:falsevacuumperturbation}: $V= n^{-1}\sum_f B_f$.

As proposed in \cite{topo_Hastings} and refined in \cite{michalakis2013stability}, a sufficient condition for stability in finitely many spatial dimensions is local topological order (LTO), which states that for any \emph{ball} region $S$ of qubits with sufficiently small radius, any operator $\OO$ supported on $S$ satisfies 
\begin{equation}\label{eq:local_indisting0}
    P_{\widetilde{S}} \OO P_{\widetilde{S}} = c_\OO P_{\widetilde{S}},
\end{equation}
for some number $c_\OO$, where $\widetilde{S}\supset S$ is a slightly larger region that includes all checks touching $S$, and $P_{\widetilde{S}}$ is a local ground state projector that projects onto the simultaneous $+1$ eigenspace of all checks inside $\widetilde{S}$. \eqref{eq:local_indisting0} is the definition of a detectable error $\mathcal{O}$ in an error-correcting code (Knill-Laflamme condition) condition by replacing $\widetilde{S}$ in the subscripts by the whole system, so \eqref{eq:local_indisting0}  can be viewed as a \emph{local} version of correctability. For stabilizer codes, \eqref{eq:local_indisting0} holds if the volume of $S$ is smaller than the code distance, and any stabilizer in $S$ can be expanded as a product of checks in $\widetilde{S}$ \cite{topo_Hastings}.   

Observe that LTO indeed rules out the Ising toric code \eqref{eq:Ising_toric}: around any ball $S$ faraway from the particular plaquette $f_0$, the local ground states $P_{\widetilde{S}}$ contain two sectors where $B_f=+1\quad (\forall f\subset \widetilde{S})$ for the first sector and $-1$ for the second. \eqref{eq:local_indisting0} is thus violated by $\OO=B_f$ that differentiates the local ground states; this operator is precisely the kind of perturbation that would make the model unstable. \eqref{eq:local_indisting0} naturally generalizes to LDPC codes, where one can still define a ball $S$ as containing all vertices of a bounded graph distance to a given vertex. This is indeed the condition used in the proof for LDPC stability on non-expander graphs  \cite{stability_maryland}. 

We now argue that LTO is not enough to get a \emph{good} bound on the stability of a phase of matter defined by an LDPC code.  (\emph{1}) Since $S$ in the local indistinguishability condition is restricted to balls, it does not control general operators that act on $>L_*+1$ sites, where $L_*$ is the maximal diameter of $S$ that satisfies \eqref{eq:local_indisting0}. Even if the original perturbation $V$ is local, such uncontrolled operators would appear at order $\sim L_*$ in perturbation series. Such operators are only suppressed by factor $\sim \ee^{-L_*}$, so that the energy splitting of the perturbed ground states is only bounded by a polynomial $1/\mathrm{poly}(n)$, as $L_*$ is bounded by the diameter of the system $\sim \log(n)$ for exponential-expanding graphs.\footnote{For subexponential-expanding graphs, \cite{stability_maryland} also only gets a quasi-polynomial bound.} This bound is much weaker than the stretched-exponential bound $\ee^{-\sqrt{n}}$ for the perturbed two-dimensional toric code \cite{topo_Hastings}.  (\emph{2}) In order for the indistinguishability diameter $L_*$ to scale as $\sim \log(n)$, the distance of the code would typically need to grow at least as a power law of $n$. Therefore the result may not hold for codes with $\mathrm{polylog}(n)$ distance, which are potentially interesting because this distance is sufficient to have a finite error-correction threshold \cite{log_distance_threshold15}. (\emph{3}) The local indistinguishability for ball regions does not directly control operators that are supported on disjoint regions. Therefore, it is not obvious how this condition can be used to prove stability under perturbations that are not geometrically local with respect to the original code. This includes important scenarios such as power-law decaying interactions in finite dimensions with sufficiently small power-law decay exponent, where the stability of the two-dimensional toric code phase has not yet been stablished.

To deal with these caveats, we propose a suitable generalization \eqref{eq:local_indisting0} to more general, possibly disconnected regions $S$, as long as its \emph{volume} (instead of diameter) is bounded by, roughly speaking, the code distance $d$. However, the naive generalization of \eqref{eq:local_indisting0} would fail: Even for the toric code, \eqref{eq:local_indisting0} does not hold for $S$ being an annulus, where a $Z$-loop operator $\OO$ in $S$ cannot be expanded to checks in $\widetilde{S}$, which is still an annulus if the ``radius'' of $S$ is larger than a constant. In this case, we need a larger region $\bar{S}$ that ``fills the hole'' of $S$, so that \eqref{eq:local_indisting0} is recovered by replacing the subscripts $\widetilde{S}$ to $\bar{S}$:
\begin{equation}\label{eq:local_indisting}
    P_{\bar{S}} \OO P_{\bar{S}} = c_\OO P_{\bar{S}}.
\end{equation}
This is because stabilizers in $S$ can be expanded into product of checks inside $\bar{S}$, as illustrated in Fig.~\ref{fig:TC}.
Note that the volume of $\bar{S}$ can be much larger than $S$; conceptually we need that the growth of $|\bar{S}|$ with respect to $|S|$ is not ``too fast", so that the condition \eqref{eq:local_indisting} is still sufficiently local.  For the toric code (Figure \ref{fig:TC}), it makes sense to identify $\bar S$ as the ``filled in" (simply connected) version of $S$, but in infinite dimensions it is not obvious how to pinpoint the ``holes'' of a region $S$ that we want $\bar{S}$ to fill. Moreover, requiring \eqref{eq:local_indisting} for any shape of $S$ would naively require one to examine exponentially many sets $S$.  

\begin{figure}
    \centering
    \includegraphics[width=0.5\linewidth]{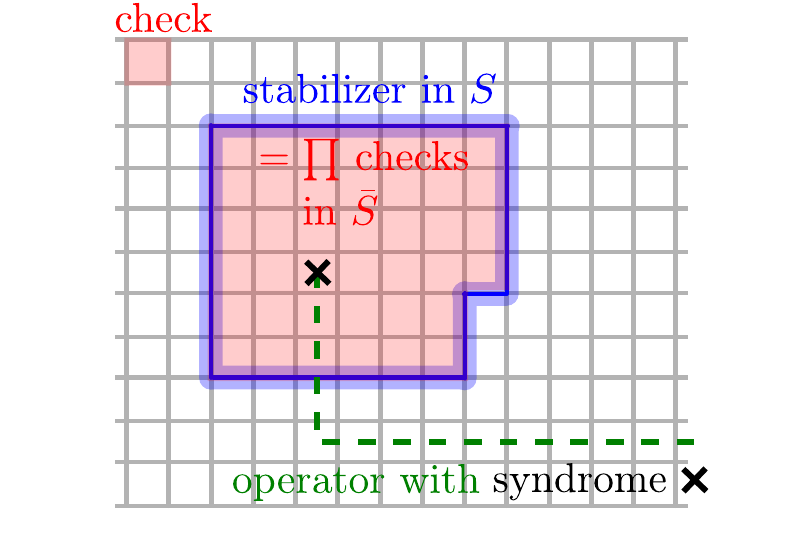}
    \caption{For the toric code, an annulus region $S$ (shaded blue region) contains a $Z$-loop stabilizer that is a product of all $Z$-checks in the shaded red region $\bar{S}$, which is the annulus $S$ with its hole filled. As a result, this $Z$-loop stabilizer may anti-commute with $X$-strings (shown by green) that has syndrome (checks that anti-commute with the operator) in $\bar{S}$ but has no syndrome in $S$. In our proof of stability for general LDPC codes, we require that $|\bar{S}|$ does not grow too quickly with $|S|$, so that in a nutshell, the stabilizer in $S$ does not anti-commutes with too many operators whose syndrome are faraway from $S$. We call this condition ``check soundness''. For example, the toric code has quadratic check soundness because $|\bar{S}|\lesssim |S|^2$ due to the area-versus-diameter relation.  }
    \label{fig:TC}
\end{figure}

Because of these challenges, we have identified a distinct condition that is sufficient for stability, which we will call \emph{check soundness}.  
\revise{
\begin{defn}\label{defn:checksoundness}
    We say a stabilizer code has $(d_{\rm c}, \mathsf{f})$-check-soundness, if any stabilizer acting on $M< d_{\rm c}$ sites is equal to a product of $\le \mathsf{f}(M)$ parity checks with an increasing function $\mathsf{f}$. 
\end{defn}

Here the weight cutoff $d_{\rm c}$ is often of the same order as code distance $d$ that scales with $n$. The more crucial parameter is the function $\mathsf{f}$, which are usually chosen as a power law $\mathsf{f}(M)\sim M^a$.\footnote{Note that here we use a different symbol $\mathsf{f}$ than the Appendix to differentiate with the previous $f$s used for faces.}
}

For the two-dimensional toric code, we have check soundness with $a=2$, because the product of parity-checks in region $R$ acts along the boundary of the region $R$, and $\text{area} \lesssim \text{perimeter}^2$ in two spatial dimensions.  

We find that $a<2$ is sufficient to prove the stability of a phase defined by a general quantum LDPC code, while in finite spatial dimensions LTO guarantees a finite $a$ that is also sufficient. Intriguingly, we show in Appendix \ref{app:constratequantum} that many constant-rate quantum LDPC codes, including ``good" codes where $d=\Theta(n)$,  have the best-possible linear check soundness $a=1$, so our stability proof applies.


\section{Main result}\label{sec:main}

For an operator $V$ that is an extensive sum of local terms, we define its local strength $\epsilon$ as, roughly speaking, \revise{the maximum $\infty$-norm ($\norm{\cdot}$) of operators in $V$ that act nontrivially on a given qubit. More precisely, \begin{equation}
    \epsilon = \max_i \sum_{S\ni i}\sum_{\bs} \norm{V_{S,\bs}}\ee^{\kappa |S|}, \label{eq:mainkappanorm}
\end{equation}
where $V=\sum_{S,\bs}V_{S,\bs}$ with $V_{S,\bs}$ acting on qubit set $S$ and having error syndrome $\mathbf{s}$ (all terms in $V_{S,\mathbf{s}}$ anticommute with the same parity-checks),  and $\kappa=\mathrm{O}(1)$ is a constant used to penalize terms of large support.  Note that having a finite $\kappa$-norm ensures that perturbations decay exponentially with \emph{volume}, not diameter, of the set $S$. See Appendix \ref{sec:prelimi} for the formal definition.} 
Such a local strength does not scale with the system size, and we desire stability against perturbations $V$ with a small but finite local strength.
We are ready to state our main result illustrated in Fig.~\ref{fig:spec}:

\begin{thm}[informal version]\label{thm:qtm}
    Let $H_0$ be the Hamiltonian of a quantum LDPC code with gap $1$ and code parameters $\llbracket n,k,d\rrbracket$. Suppose \begin{equation}
        d=\Omega(\log n),
    \end{equation} 
    and $H_0$ has better-than-quadratic check soundness $a<2$. There exists a positive constant $\epsilon_0$, such that for any perturbation $V$ whose local strength $\epsilon$ is small: $\epsilon<\epsilon_0$,
    the perturbed Hamiltonian $H:=H_0+V$ still has $2^k$ low-energy states, separated by a gap $\ge 1/2$ from the remainder of the spectrum.  These lowest-energy states are almost degenerate with small energy difference \begin{equation}\label{eq:dE=ed}
        \delta E=\epsilon\, \ee^{-\Omega(d)}.
    \end{equation}
    Moreover, the perturbed lowest-energy states are connected to the unperturbed ground states by a quasi-local unitary whose generator has small $\mathrm{O}(\epsilon)$ local strength.

    The above conclusions hold for the known constructions of good quantum LDPC codes \cite{Panteleev_2022,qTanner_codes,dinur2022good} with $d=\Theta(n)$, quantum expander codes from hypergraph product \cite{HGP,qExpander15} with $d=\Theta(\sqrt{n})$, and finite-dimensional codes with local topological order\footnote{For finite dimensions, this requires that (as we show later) we can relax the check soundness condition to be any polynomial.}.
\end{thm}

See Appendix \ref{sec:main_proof} for the formal statements and the proof, where we employ the Schrieffer-Wolff transformation (SWT) scheme used to prove stability of finite-dimensional codes \cite{topo_Hastings}. This is in contrast to the scheme based on quasi-adiabatic evolutions \cite{bravyi2011short,michalakis2013stability,stability_maryland}. We do not employ the latter scheme because it inevitably invokes Lieb-Robinson bounds \cite{ourreview} for Hamiltonian evolution whose locality is exponential-decaying with the \emph{diameter}, not the \emph{volume}, which seems to only produce a quasi-polynomial bound on energy splitting \cite{stability_maryland} for expander graphs.

As we do not require the perturbation $V$ to be geometrically local, Theorem \ref{thm:qtm} proves stability against general $q$-local perturbations, as long as $   q=\Theta(1)$ and the local perturbation strength $\epsilon$ that each qubit feels is smaller than some constant. In particular, this holds for finite-dimensional codes with local topological order, and significantly improves upon the previous result \cite{michalakis2013stability}, which proves stability with a polynomial energy splitting $\delta E=\mathrm{O}(n^{-\alpha})$ for power-law decaying perturbations with a sufficiently fast decay exponent. Here we achieve exponential energy splitting, e.g. $\delta E=\ee^{-\Omega(\sqrt{n})}$ for the toric code, for all power-law decaying exponents. Note that \cite{michalakis2013stability} does work with the weaker assumption that $H_0$ is frustration-free, rather than requiring it to come from a stabilizer code, as in this paper.

\subsection{\revise{Proof sketch for a simple example}}
Here we sketch the proof of Theorem \ref{thm:qtm} using a simple example, and point out our technical ingredients along the way. Consider a non-geometrically-local perturbation on a magnetic-field Hamiltonian $H_0$: \begin{equation}\label{eq:SWT_intro_Ising}
    H_0=-\frac{1}{2}\sum_{i=1}^n Z_i,\quad V_1= \frac{1}{n} \sum_{1\le i<j\le n}\epsilon_{ij} (X_i X_j+Z_iZ_j),
\end{equation}
where $\abs{\epsilon_{ij}}\le \epsilon$. Note that we have normalized the perturbation $V_1$ so that its local strength acting on any one qubit is bounded by $\epsilon$.

$H_0$ has a single gapped ground state $\ket{\bm{0}}$, the all-zero state, so can be viewed as a trivial $\llbracket n,0,\infty \rrbracket$ quantum code. Since $H_0$ satisfies LTO in a trivial way, and can be viewed as in one spatial dimension, it is stable for small $\epsilon$ according to Theorem \ref{thm:qtm}, which goes beyond previous stability proofs for topological orders due to the non-geometrically-local perturbation. An alternative way of applying Theorem \ref{thm:qtm} to this example is to observe $H_0$ has the best-possible linear check soundness, because the checks $Z_i$ just act on different qubits. 

$V_1$ couples the ground state $\ket{\bm{0}}$ of $H_0$ to excited states. We wish to find a unitary $U$ that block-diagonalizes the Hamiltonian $H=H_0+V_1$ such that $\ket{\bm{0}}$ is an eigenstate of the dressed Hamiltonian $U^\dagger H U$. Here the two blocks correspond to $\ket{\bm{0}}$ and its orthogonal subspace. However, such a unitary $U$ is hard to find exactly, and may be very nonlocal. 

Instead, one can first aim something weaker: to find a unitary $U_1$ that block-diagonalizes $H$ up to second order: \begin{equation}\label{eq:SWT_intro}
    U_1^\dagger H U_1=(H_0+D_2) + V_2,
\end{equation}
where $D_2\ket{\bm{0}}\propto \ket{\bm{0}}$ and $V_2\propto \epsilon^2$. Note that $V_2$ is still an extensive sum of many local terms so $\norm{V_2}\sim \epsilon^2 n\gg 1$. For each local term $V_{1;ij}\propto X_iX_j+Z_iZ_j$ of $V_1$, it is not hard to find a unitary $U_{1;ij}$ acting on sites $i,j$ that rotates away the off-block-diagonal part $X_iX_j$ of $V_{1;ij}$, such that \eqref{eq:SWT_intro} holds for $U_1=U_{1;ij}$ if $V_1$ only consists of this single term. Moreover, $U_{1;ij}=\ee^{A_{1;ij}}$ for an anti-Hermitian operator $A_{1;ij}$ supported on $i$ and $j$ that is as small as $V_{1;ij}$: $A_{1;ij}\sim \epsilon/n$. As a result, even if $V_1$ contains many local terms $V_{1;ij}$, we have $A_{1;ij}$ for each term so that we can define the SWT unitary $U_1=\ee^{A_1}$ with $A_1=\sum_{1\le i<j\le n} A_{1;ij}$. This $U_1$ is quasi-local in the sense that $A_1$ has small local strength $\epsilon$. On the other hand, this $U_1$ achieves \eqref{eq:SWT_intro} because at first order every $V_{1;ij}$ is independently rotated by the corresponding $A_{1;ij}$ in $A_1$. This SWT procedure can be understood as a nondegenerate perturbation theory applied to local operators.  

Since $U_1$ is quasi-local, $D_2,V_2$ are also expected to be sufficiently local: For example, $V_2$ is roughly $[A_1,V_1]$ so a typical local term of it acts on only $\le 3$ sites, as $A_1,V_1$ are both $2$-local. However, $V_2$ strictly speaking contains Heisenberg-evolved operators like $\ee^{-A_1}V_1\ee^{A_1}$, so in principle contains nonlocal operators that act on the whole system. We thus need to control the locality of $V_2$ by examining how a term in it decays with the support (or operator size) of that term. Since the operators like $A_1$ are non-geometrically-local, we cannot invoke Lieb-Robinson bounds. Instead, we use locality bounds obtained from cluster expansion techniques \cite{abanin2017rigorous,our_metastable}. This enables us to show that, a term in $V_2$ decays exponentially with its support volume, so $V_2$ is still a sufficiently local operator dominated by $q$-local terms with small $q\approx 3$ (albeit non-geometrically-local). We point out that this exponential decay with volume (instead of e.g. diameter of the support) is crucial to obtain the final bound \eqref{eq:dE=ed} that is exponential in the code distance. 

\revise{We refer to Lemma \ref{lem:clustexp} in the Appendix for the mathematical cluster expansion bound we use; here we provide a rough intuition on how it works: For a local operator $\OO$ at site $i_0$, one can Taylor-expand $\ee^{-A_1}\OO \ee^{A_1}=\OO - [A_1, \OO] + \cdots$ and then replace each $A_1$ by the sum of its local terms. Then $\ee^{-A_1}\OO \ee^{A_1}$ will be a sum over terms $[\cdots[A_{1;i'j'},[A_{1;ij},\OO]]\cdots]$ organized by their ``clusters'', i.e. the supports $(\cdots,i'j',ij,i_0)$. This expansion clearly has the property we want that each cluster has an operator norm that decays exponentially with the cluster volume, because each $A_{1;ij}$ is small. The technical part is then to show this expansion converges when truncating the sum at some cutoff cluster volume. The idea is that a cluster can be viewed as appending/growing one extra $A_{1;ij}$ on a smaller cluster; the cluster expansion technique establishes a condition for each local growth, under which the whole expansion would indeed converge. This condition translates to a small local strength of $A_1$ in our context, which is satisfied from discussions above. }

Observe that we have achieved \eqref{eq:SWT_intro} that is already an order of improvement over the original problem where $V_1\propto \epsilon$. Moreover, the dressed Hamiltonian in \eqref{eq:SWT_intro} looks almost exactly like the original problem, where now we are perturbing $H_0+D_2$ by $V_2$. We then want to follow the procedure of obtaining $U_1$ to suppress the perturbation $V_2$ to an even higher order so that the Hamiltonian becomes \begin{equation}\label{eq:SWT_intro_3}
    U_2^\dagger U_1^\dagger H U_1 U_2=(H_0+D_3) + V_3,
\end{equation}
where $U_2$ is another quasi-local SWT unitary. However, the dressed Hamiltonian in \eqref{eq:SWT_intro} is \emph{not} exactly the original problem because of the extra $D_2$ that makes the major part $H_0+D_2$ unsolvable: This step thus requires extra conditions of $H_0$ to guarantee that $H_0+D_2$ is, roughly speaking, still a ``code'' such that there remains a $\Omega(1)$ gap above $\ket{\bm{0}}$. 

Here $D_2=n^{-1}\sum_{ij}\epsilon_{ij}Z_iZ_j$ is the block-diagonal part of $V_1$. Although $\norm{D_2}\gg 1$, it cannot close the gap of $H_0+D_2$. To see this, consider a particular state $\ket{\psi}=\ket{\bm{0}}_{1,\cdots,n/2}\otimes \ket{\bm{1}}_{n/2+1,\cdots,n}$ (supposing $n$ is even) and a particular choice of $\epsilon_{ij}=\epsilon\, (i\le n/2,j\ge n/2)$ and $\epsilon_{ij}=0$ otherwise. This choice makes $\ket{\psi}$ \emph{gain} a large extensive energy from $D_2$ comparing to $\ket{\bm{0}}$: \begin{equation}\label{eq:sketch_E_gain}
    \bra{\bm{0}}D_2\ket{\bm{0}} - \bra{\psi}D_2\ket{\psi} = \epsilon n/2 \gg 1,
\end{equation}
much larger than the gap of $H_0$.
However, this does not make $\ket{\psi}$ the new ground state: it has an extensive energy \emph{penalty} $\bra{\psi}H_0\ket{\psi} -\bra{\bm{0}}H_0\ket{\bm{0}} = n/2$ from the original Hamiltonian $H_0$, which is larger than \eqref{eq:sketch_E_gain}!

The gap of $H_0+D_2$ can be attributed to the check soundness property of $H_0$, and this turns out to be a general sufficient condition for stability: Intuitively, $D_2$ contains small stabilizers, and check soundness implies that each stabilizer $\mathcal{S}$ in $D_2$ can be expanded to a product of a small number of checks $C_1,\cdots,C_M$ in $H_0$. As a result, in order for a state to gain energy (relative to $\ket{\bm{0}}$) from stabilizer $\mathcal{S}$, it has to flip one of the few checks in $C_1,\cdots,C_M$. As $M$ is small from check soundness, this means that when considering all stabilizers in $D_2$, a sufficiently large fraction $\sim 1/M$ of the checks in $H_0$ need to be flipped, so that the energy penalty from $H_0$ always wins and the gap is preserved at small but constant $\epsilon$.

Having figured out how to do a single step of SWT using check soundness,
the goal is then to repeat the SWT procedures \eqref{eq:SWT_intro} and \eqref{eq:SWT_intro_3} to very high orders: We do not need to compute $D_m,V_m,U_m$ explicitly, but we need to control their locality as they become more nonlocal at higher orders. With such technicalities addressed by the cluster expansion bounds, we finally arrive at a sufficiently high order $m_*\sim n$ for the example \eqref{eq:SWT_intro_Ising}, so that the dressed Hamiltonian becomes \begin{equation}\label{eq:UHU=goal}
    U^\dagger H_0 U=H_0+D_{m_*}+V_{m_*},
\end{equation}
where $U=U_1U_2\cdots U_{m_*-1}$ is still quasi-local, $H_0+D_{m_*}$ has $\ket{\bm{0}}$ as its gapped ground state from check soundness, and $V_{m_*}\sim \epsilon^n$ is so small that does not close the gap and only causes an exponentially small energy splitting (if the model has multiple ground states unlike \eqref{eq:SWT_intro_Ising}). This establishes the stability of the original problem \eqref{eq:SWT_intro_Ising}.

\subsection{\revise{High-level summary of the proof of Theorem \ref{thm:qtm}}}
\revise{Our goal is to find a quasi-local unitary $U$ that achieves \eqref{eq:UHU=goal}, where $H_0$ is the Hamiltonian of an LDPC code with check soundness, and $D_{m_*},V_{m_*}$ are still quasi-local Hamiltonians.  We desire $D_{m_*}$ to be block-diagonal between the codespace and the rest of the Hilbert space, and that $H_0+D_{m_*}$ retains the codespace as a gapped $2^k$-fold degenerate ground state subspace.  Lastly we require \begin{equation}\label{eq:V<exp_main}
    \norm{V_{m_*}}=\ee^{-\Omega(d)}.
\end{equation} 
The existence of such a $U$ immediately implies Theorem \ref{thm:qtm}: Since $V_{m_*}$ is exponentially small, the spectrum of $H$ agrees with that of $H_0+D_{m_*}$ up to energy differences given by \eqref{eq:dE=ed}, while $H_0+D_{m_*}$ has all desired properties, including a preserved $\Omega(1)$ gap. 
To establish all of these facts, there are three crucial steps: \begin{enumerate}
    \item To find $U$, we formally perform a perturbation theory in $\epsilon$ by iterated SWTs, up to order $m_*\sim d$.    $U=U_1U_2\cdots U_{m_*}$, where at each order $U_l=\exp[A_l]$ where $A_l$ is chosen carefully by using the decomposition \eqref{eq:mainkappanorm}.   Terms $V_{S,\mathbf{s}}$ which violate more parity-checks can be removed by terms in $A_l$ with smaller coefficients.
    \item We then prove the convergence of this series up to the desired order: in particular, that $D_{m_*}$ and $V_{m_*}$ are quasilocal in the sense of \eqref{eq:mainkappanorm} with a finite $\kappa$.  The cluster  expansion proves useful here, along with check soundness, which constrains the number of non-vanishing possible terms in the cluster expansion.
    \item  Lastly, we verify that $H_0+D_l$ remains gapped at each step of the SWT.  Check soundness is again crucial to ensure that this gap cannot close, by forbidding pathological examples such as \eqref{eq:Ising_toric}.
\end{enumerate}
}

\section{Application to classical codes and quasi-adiabatic evolution}


Our result also generalizes to prove (in Appendix \ref{app:classical}) stability of classical LDPC codes with check soundness: 

\begin{cor}[Informal version]\label{cor:cla}
    Theorem \ref{thm:qtm} also holds for $H_0$ being a classical LDPC code, as long as each term in perturbation $V$ is symmetric under conjugation by any logical operator (made out of Pauli $X$s), as discussed around \eqref{eq:isingxonly}.
\end{cor}
Here we need to restrict to symmetric perturbations, because otherwise the gap is in general unstable to $Z$-field perturbations that locally differentiate the ground states, lifting some of the ground states to false vacuum states \cite{coleman,our_preth,our_metastable} in the spectrum of excited states, as previously discussed. There do exist classical codes with check soundness (see Appendix \ref{app:constratequantum}) to which this result applies.

It would be interesting to compare this ground-state stability result to the provably robust eigenstate localization in classical LDPC codes on expander graphs \cite{MBL_LDPC24}, which can be viewed as a weaker sense of stability, albeit one that holds for all states in the spectrum with a bounded energy density. The stability theorem of \cite{MBL_LDPC24} is weaker in the sense that, the eigenstates form groups where different groups corresponding to different codewords do not mix with each other; perturbation can cause strong mixing within groups.  The guarantee of Corollary \ref{cor:cla} is closer, in spirit, to the conjecture for ``l-bits" -- namely, that there would exist \emph{local unitary} rotations \cite{MBLrev_Rahul} from a product state into a many-body localized eigenstate.   Of course, Corollary \ref{cor:cla} only holds for the ground states, whereas the localization guarantee of \cite{MBL_LDPC24} also holds for excited states.

A limitation of Corollary \ref{cor:cla} is that it does not apply to classical codes without check soundness. An important example of a classical code without check soundness is the Ising model \eqref{eq:IsingH0}.
If the shortest path in the interaction graph from site $1$ to $M$ is $1\rightarrow 2\rightarrow \cdots \rightarrow M$, then stabilizer $Z_1 Z_M$ expands to at least $M-1$ checks $Z_1Z_2,Z_2Z_3,\cdots,Z_{M-1}Z_M$ along the path. This lack of check soundness is physical: The Ising model is known to be unstable against long-range $Z_iZ_j$ perturbations \cite{Ising_unstab80,Ising_unstab81,longrange_stab23}. The idea is that although $Z_iZ_j$ is globally symmetric, it ``fractionalizes'' into two parts $Z_i$ and $Z_j$ that each is charged under the symmetry and locally destroys the error-correction phase. Intriguingly, a similar instability also exists for symmetry-protected topological states \cite{phase_nonlocal24}. 

Nevertheless, we can generalize our result as follows, focusing on the Ising model: 
\begin{cor}[Informal version]\label{cor:Ising}
    The conclusions of Theorem \ref{thm:qtm} hold for $H_0$ being the Ising model \eqref{eq:IsingH0} on any graph, so long as the perturbation $V$ is a sum of sufficiently local terms which act on connected subsets of qubits. 
\end{cor}
Here we crucially restrict to \emph{connected} sets, ruling out long-range perturbations that decay slower than exponential with respect to the connected support.  Corollary \ref{cor:Ising} improves previous results \cite{bravyi2011short,longrange_stab23} in that (\emph{1}) we do not require finite dimensions; (\emph{2}) we obtain an energy splitting exponentially small in the system \emph{volume} $\delta E=\ee^{-\Omega(n)}$ instead of diameter; (\emph{3}) we find a quasi-local unitary $U=\mathcal{T}\exp\mlr{\int^1_0 \dd t A(t)}$ connecting the old and new ground states whose generator $A(t)=\sum_S A_S(t)$ decomposed to connected sets decays exponentially with volume $|S|$ instead of diameter: \begin{equation}\label{eq:AS<expS}
    \norm{A_S(t)}=\epsilon\, \ee^{-\Omega(|S|)}.
\end{equation}

\revise{To understand this improved bound, we recall that}
when two gapped Hamiltonians $H_1,H_2$ are in the same phase, i.e. there exists a smooth path of Hamiltonians $\{H(s): s\in [0,1],H(0)=H_1,H(1)=H_2\}$ connecting the two without closing the gap, a quasi-local unitary $U$ is known to exist that connects the two gapped subspaces \cite{quasiadiabatic05,bachmann2012automorphic}, which is called quasi-adiabatic evolution (continuation). However, even if $H(s)$ is strictly local (i.e. interactions have finite-range) in finite dimensions, the generator $A(t)$ of the quasi-adiabatic evolution is only known to decay as \cite{bachmann2012automorphic}: \begin{equation}
    \norm{A_S(t)} = \exp\mlr{-\Omega\lr{\frac{r}{\log^2 r}}},
\end{equation}
where $r=\diam S:=\max_{i,j\in S}\mathsf{d}(i,j)$ is the diameter of $S$. Our results Theorem \ref{thm:qtm} and Corollary \ref{cor:Ising} significantly improves this almost exponential decay with the diameter, to an exponential decay with the volume $|S|$ in \eqref{eq:AS<expS}, with a restriction that the Hamiltonians $H_1,H_2$ are close enough to a solvable fixed point $H_0$ as an LDPC code. This improved tail bound enables us to show much stronger symmetry properties for gapped ground states around fixed points in the phase space, which are discussed elsewhere \cite{volumeLaw_disorder24}.

\revise{Beyond classical codes under symmetric perturbation, \eqref{eq:AS<expS} also holds for quantum code Hamiltonians under general perturbation.  This fact follows from Theorem \ref{thm:qtm}.}

\section{Conclusion}
We have established that a broad family of quantum LDPC codes define stable gapped phases of matter, whose (approximate) ground state degeneracy is robust to perturbation.  Under sufficiently small perturbations, quasilocal unitaries can relate one ground state to another in the same phase.   We have further argued that check soundness, rather than local topological order, is the critical property that renders a quantum code a stable phase of matter.   LDPC codes then represent a setting in which certain mathematical definitions about what constitutes a robust phase of matter can be extended into ``infinite dimensions", beyond the paradigms of Landau theory or statistical/quantum field theory.

These results represent a valuable contribution to a rapidly-growing body of literature \cite{Rakovszky:2023fng,Rakovszky:2024iks,Hong:2024fsg,MBL_LDPC24,stability_maryland,stability_Khemani24} on the \emph{physics} of quantum LDPC codes.  Like classical LDPC codes \cite{Montanari_2006}, it increasingly appears that quantum LDPC codes form a kind of quantum glass with ergodicity-breaking at low temperature, perhaps related to the complex entanglement (NLTS) of the low-energy states \cite{Anshu_2023}. In particular, our work combined with \cite{Hong:2024fsg} suggests that these quantum LDPC codes are \emph{stable} self-correcting quantum memories, where both coupling to a low-temperature thermal bath and perturbing the Hamiltonian can not corrupt the encoded quantum information. These unusual phases of matter break our intuition about thermodynamics in finitely-many spatial dimensions, including the existence of a third law of thermodynamics: namely, that any well-defined quantum \emph{phase} should have a vanishing entropy density at zero temperature.  We expect that this new family of phases of matter will continue to push the limits of our understanding of the foundations of statistical mechanics, and may indeed represent an ideal setting to ``stress test" the lore of quantum statistical mechanics that was developed with our three-dimensional universe in mind.

\emph{Note Added---}  As our manuscript was finalized, we became aware of related work \cite{vedikanew}.  Our preprints appeared appear in the same arXiv posting.



\section*{Acknowledgements}
We thank Yifan Hong for many helpful discussions on error correcting codes, and especially for alerting us to known results on quantum LDPC codes that are equivalent to ``check soundness".  This work was supported by the Department of Energy under Quantum Pathfinder Grant DE-SC0024324. 

\begin{appendix}
\renewcommand{\thesubsection}{\thesection.\arabic{subsection}}
\renewcommand{\thesubsubsection}{\thesubsection.\arabic{subsubsection}}

\section{Preliminaries}\label{sec:prelimi}

\subsection{Graphs}

We define the floor function $\lfloor x \rfloor$ as the largest integer that is not larger than $x$.

Let $\Lambda$ denote a set of vertices on an undirected graph with edge set $\mathcal{E}$.  When discussing subsets $S\subset\Lambda$, the complement of $S$ will be denoted as $S^{\rm c}$.  
 $S\setminus S'$ denotes the set $\{i\in S: i\notin S'\}$.   A set $S$ is connected if for any two $i\in \Lambda$ and $j\in \Lambda$, there exists a subset of $\mathcal{E}$ of the form $(i, k_1), (k_1, k_2), \ldots, (k_l, j)$ such that $\lbrace k_1,\ldots, k_l\rbrace \subseteq \Lambda$; otherwise, $S$ is disconnected.  We assume that $\Lambda$ is connected.

 \begin{defn}[Distances on graphs]
     For each pair of vertices $i,j\in \Lambda$, the distance $\mathsf{d}(i,j) \in \mathbb{Z}_{\ge 0}$ between them is the length of the shortest path in graph $(\Lambda,\mathcal{E})$ between them, with $\mathsf{d}(i,j)=0$ if and only if $i=j$.   We similarly define the distance from a vertex to a set $\mathsf{d}(i,S):=\min_{j\in S} \mathsf{d}(i,j)$, and between two sets $\mathsf{d}(S,S'):=\min_{i\in S, j\in S'} \mathsf{d}(i,j)$. 
 \end{defn}

 \begin{defn}[Ball around a vertex]
     A ball $B_{i,r}$ of center $i$ and radius $r$ is the set of vertices that are of distance $\le r$ to $i$.  Let $\Gamma_i(r):= |B_{i,r}|$ be the upper bound on the volume of a radius-$r$ ball, and $\gamma_i(r):= |B_{i,r}|-|B_{i,r-1}|$ be the bound on the surface volume.  We often drop the subscript $\Gamma_i \rightarrow \Gamma$ when $i$ is clear from context.
 \end{defn}

 \begin{defn}[Degree and dimensionality] \label{def:degreedimensionality}
     The \textbf{degree} of a vertex $i$ is $\gamma_i(1)$.  If a graph is ``constant-degree" with $\gamma_i(1) \le \Delta$ for all $i$, then we always have a bound \begin{equation}\label{eq:Gamma<Delta}
    \gamma(r)\le \Gamma(r) \le \Delta^r.
\end{equation}
If $\Gamma(r) \ge \exp[\alpha r]$ for some finite $\alpha>0$, up to $r = \mathrm{O}(\log n)$, we say that the graph is an \textbf{expander graph}.

If there exist constants $\mathcal{D},c_{\mathcal{D}}$ such that \begin{equation}\label{eq:Gamma<finite_d}
    \Gamma(r)\le c_{\mathcal{D}} r^{\mathcal{D}},\quad \gamma(r)\le c_{\mathcal{D}} r^{\mathcal{D}-1},\quad \forall r\ge 1,
\end{equation}
we say the graph $G$ is of finite spatial dimension $\mathcal{D}$.  Note that for any connected graph we require $\mathcal{D} \ge 1$.
 \end{defn}



 

\subsection{Quantum stabilizer codes}
In this paper, we study quantum many-body systems consisting of $n$ qubits, i.e. those whose Hilbert space \begin{equation}
    \mathcal{H} := \left(\mathbb{C}^2\right)^{\otimes n}.
\end{equation}
We will usually denote $\Lambda := \lbrace 1,\ldots, n\rbrace$ to be the set of qubits, and $i\in\Lambda$ to denote one particular qubit.  We let $X_i$, $Z_i$ and $Y_i = \mathrm{i}Z_iX_i$ denote the three Pauli matrices acting on qubit $i$.  

We now restrict our study further to $H$ which arise from quantum (Pauli) stabilizer codes.  We are deliberately restricting our discussion of quantum codes to only the bare minimal ingredients necessary to understand the proof of our main result.
\begin{defn}[Stabilizer code]\label{def:stabilizer_code}
    Let $\Lambda_{\mathrm{c}}$ denote a set of \textbf{parity checks} (checks for short):   \textbf{Pauli strings} of the form $\prod_{j\in S} A_j$ where each $A_j \in \lbrace X_j, Y_j, Z_j\rbrace$ and $S\subset \Lambda$, which obey \begin{equation}
        [\mathcal{Q}_1,\mathcal{Q}_2] = 0, \;\;\; \text{ if } \mathcal{Q}_{1,2} \in \Lambda_{\mathrm{c}}.
    \end{equation}
    If for some $\mathcal{Q}_1,\ldots, \mathcal{Q}_n\in \Lambda_{\mathrm{c}}$, $\prod_{a=1}^n \mathcal{Q}_a = I$ (the identity), the stabilizer code is \textbf{redundant}.  In what follows, we will usually refer to checks $\mathcal{Q}_C$ by the subset $C\subset\Lambda$ on which they act; $|C|$ refers to the length of the Pauli string $\mathcal{Q}_C$.  A product of parity checks is called a \textbf{stabilizer}.

    A \textbf{logical operator} is a Pauli string which commutes with all stabilizers, yet is not a product of elements of $\Lambda_{\mathrm{c}}$.  A \textbf{quantum code} has $2k$ logical operators whose algebra is isomorphic to that \revise{generated by} $X_1,Z_1,\ldots, X_k, Z_k$.  The length of the smallest non-identity logical operator is called the \textbf{code distance} $d$.  If $|\Lambda|=n$, we call this an $\llbracket n,k,d\rrbracket$ code.

    A \textbf{low-density parity-check (LDPC) code} is a stabilizer code for which: (1) any $\mathcal{Q}_C \in \Lambda_{\mathrm{c}}$ has the property that $|C| \le q$, for some O(1) value $q$, and (2) for any $i\in\Lambda$, $|\lbrace \mathcal{Q}_C \in \Lambda_{\mathrm{c}} : i \in C\rbrace | \le q^\prime$ for O(1) $q^\prime$.  In other words, each qubit is involved in a finite number of checks, and each check involves a finite number of qubits.  

    The \textbf{Tanner graph} of a stabilizer code is a bipartite graph with vertices $\Lambda\cup \Lambda_{\rm c}$, where there is an edge $(i,\mathcal{Q}_C)$ with $i\in \Lambda,\mathcal{Q}_C\in \Lambda_{\rm c}$ if and only if $i\in C$.

    A \textbf{CSS code} is a stabilizer code where any $\mathcal{Q}_C \in \Lambda_{\mathrm{c}}$ is a product of all Pauli $X$s or all Pauli $Z$s.  $k$ logical operators will be products of all $X$s, and $k$ will be products of all $Z$s.
    
    A \textbf{classical code} is a stabilizer code where any $\mathcal{Q}_C$ is a product of Pauli $Z$s.  In a classical code, we restrict study to logical operators that are  a product of Pauli $X$s. 
\end{defn}

For any CSS code, all of its $Z$ checks in $\Lambda_{\rm c}$ form the checks of a classical code. The same holds for $X$ checks with a local basis change $X_i\leftrightarrow Z_i$ for each qubit.

\begin{defn}[Code Hamiltonian]
    Given an LDPC code, we define an unperturbed \textbf{code Hamiltonian} \begin{equation}\label{eq:H0=}
    H_0 = \sum_{\mathcal{Q}\in \Lambda_{\rm c}} \lambda_{\mathcal{Q}} \frac{I-\mathcal{Q}}{2},
\end{equation}
where, for some finite $a<\infty$, \begin{equation}
    1 \le \lambda_{\mathcal{Q}} \le a. \label{eq:minlambdacheck}
\end{equation}
We define a \textbf{codespace (ground state subspace) projector} \begin{equation}
    P:=\prod_{\mathcal{Q}\in \Lambda_{\rm c}}\frac{I+\mathcal{Q}}{2}.
\end{equation}
$|\psi\rangle$ is a ground state of $H_0$, obeying $H_0|\psi\rangle = 0$, if and only if $P|\psi\rangle =|\psi\rangle$.  There are $2^k$ linearly independent ground states.  Notice that $H_0$ has a spectral gap $\ge 1$ between the ground state and excited states.
\end{defn}



It will often be useful to call upon a \emph{local} version of the ground state projector $P$, restricted to set $S\subset \Lambda$:
\begin{equation}
    P_S := \prod_{\mathcal{Q}_C\in \Lambda_{\rm c}: C\subset S}\frac{I+\mathcal{Q}_C}{2} .
\end{equation}
Similarly, we define the local projector onto excited states $Q_S:=I-P_S$, the global excited state projector $Q := I-P$, and the local patch of the Hamiltonian in set $S$: \begin{equation}
    H_S:=\sum_{\mathcal{Q}_C\in \Lambda_{\rm c}: C\subset S} \lambda_{\mathcal{Q}_C} \frac{I-\mathcal{Q}_C}{2}.
\end{equation}

\begin{defn}[Edge set for the code graph]
    Given a quantum code, we define the edge set $\mathcal{E}$ such that \begin{equation}
        \mathcal{E} := \lbrace (i,j) = (j,i) : \exists \mathcal{Q}_C \in \Lambda_{\mathrm{c}}  \text{ with } \lbrace i,j\rbrace \subseteq C \rbrace.
    \end{equation}
    Notice that the code Hamiltonian $H_0$ thus couples $i \ne j$ if and only if $\mathsf{d}(i,j) = 1$.
\end{defn}

Given this convention for the edge set on the interaction graph of the code Hamiltonian, notice that the parameter $\Delta$ from Definition~\ref{def:degreedimensionality} obeys $\Delta \le qq^\prime - 1$.\footnote{In the proof of Theorem \ref{thm:qtm}, we will often substitute $q$ or $q^\prime$ for $\Delta$ out of convenience.} 


\subsection{Operator decomposition into syndromes}
\begin{defn}[Syndromes]
    Let $\OO$ be a Pauli string.  We write $\OO = \OO_{\mathbf{s}}$ where the \textbf{syndrome} $\mathbf{s} \in \mathbb{F}_2^{\Lambda_{\mathrm{c}}}$ is defined such that \begin{equation}
    \OO_{\bs}\mathcal{Q}=\left\{\begin{array}{cc}
       -\mathcal{Q}\OO_{\bs}, & \mathcal{Q}\in \bs    \\
       \mathcal{Q}\OO_{\bs}, & \mathcal{Q}\notin \bs  
    \end{array}\right. .
\end{equation}
  Notice that if the code is redundant, not all syndromes $\mathbf{s}$ are possible. More generally, we write $\OO = \OO_{\mathbf{s}}$ if operator $\OO$ can be decomposed as a sum of Pauli strings all having the same syndrome $\mathbf{s}$.
  \end{defn}

  \begin{defn}[Strongly supported operators]
  For any operator $\OO$, we can decompose it as 
    \begin{equation} \label{eq:O=OSs}
    \OO = \sum_{S,\bs} \OO_{S,\bs},
\end{equation}
      where each $\OO_{S,\bs}$ has syndrome $\mathbf{s}$ and is \textbf{strongly supported} in $S$. More precisely, we demand that $\OO_{S,\bs}$ acts trivially on $S^{\rm c}$, and any check that does not commute with $\OO_{S,\bs}$ is contained in $S$: \begin{equation}\label{eq:HcommV}
    [\mathcal{Q}_\chec, \OO_{S,\bs}] = 0,\quad \mathrm{if}\quad \chec\not\subset S.
\end{equation}
  \end{defn}

By definition, if an operator is strongly supported in $S$, it is supported in $S$, i.e. it does not act on $S^{\rm c}$. On the other hand, if an operator is supported in $S'$, it is strongly supported in a larger set $S$ that includes $S'$ and its neighborhood sites that are connected to $S'$ by a check.

$\mathcal{O}_{S,\mathbf{s}}$ can anticommute with operators which flip syndromes $\mathbf{s}$ whose checks are not contained in $S$.  An example using the toric code was stressed in the main text (Figure \ref{fig:TC}), and accounting for this possibility will be an important issue in our proof. We also stress that $\OO_{S,\mathbf{s}}$ is not in general a simple Pauli string, and in general it can include arbitrary functions of stabilizers contained within $S$. When the syndrome is not important, we will sometimes also use \begin{equation}
    \OO = \sum_S \OO_S,
\end{equation}
to simplify notation, where $\OO_S:=\sum_\bs \OO_{S,\bs}$.



 \begin{defn}[$\kappa$-norm \revise{(local strength in the main text)}]
     \revise{For any $\kappa>0$,} we define a $\kappa$-norm \begin{equation}\label{eq:kappanorm_def}
    \norm{\OO}_\kappa := \max_i \sum_{S\ni i}\sum_{\bs} \norm{\OO_{S,\bs}} \ee^{\kappa |S|}.
\end{equation}
We will frequently use that \begin{equation}
    \norm{\OO}_{\kappa'}\le \norm{\OO}_\kappa,\quad \forall \kappa'\le \kappa.
\end{equation}
 \end{defn}

Since one can expand any operator into Pauli strings where each Pauli has a well-defined syndrome, the decomposition \eqref{eq:O=OSs} always exists; furthermore, locality is well preserved in this procedure: 
\begin{prop}[Adapted from Proposition 18 in \cite{our_metastable}]\label{prop:decomp}
For any operator $\OO=\sum_{S} \OO_S$ where $\OO_S$ is supported in $S$ (not necessarily strongly supported) with \begin{equation}\label{eq:OS<1}
    \max_i \sum_{S\ni i} \norm{\OO_S}\ee^{\kappa|S|}\le 1,
\end{equation}
one can always decompose it to the form \eqref{eq:O=OSs}, which induces local norm $\norm{\OO}_{\kappa'}\le c_\Delta$ with $\kappa'=c_\Delta'\kappa-\log(2)\Delta$, where the constants $c_\Delta,c_\Delta'$ are determined by $\Delta$.
\end{prop}

\revise{
In words, it is always easy to decompose an extensive operator $\OO$ into the form $\sum_{S} \OO_S$ such that the corresponding local strength of this decomposition (left hand side of \eqref{eq:OS<1}) is bounded by a constant ($1$ in \eqref{eq:OS<1} by normalizing the operator). Proposition \ref{prop:decomp} ensures that one can furthermore find a decomposition $\OO_{S,\bs}$ that keeps track of the syndrome information, which induces a finite $\kappa$-norm defined in \eqref{eq:kappanorm_def} as long as the original $\OO_S$ decays exponentially with $|S|$ in a sufficiently fast way.
}

\section{Check soundness}\label{sec:check_sound}
A critical property of a quantum stabilizer code will be the existence of check soundness.  The main purpose of this appendix is to collect known properties about check soundness, which are usually expressed in a different form in the existing coding literature. 
\subsection{Formal definition}
\begin{repdefn}{defn:checksoundness}[Check soundness]\label{def:check_exp}
    We say a stabilizer code has $(d_{\rm c}, f)$-check-soundness, if for any (Pauli string) stabilizer $\OO$, obeying $[H_0, \OO]=0$, acting on $M< d_{\rm c}$ sites, there exist $M'\le f(M)$ parity checks $\mathcal{Q}^{(1)},\cdots,\mathcal{Q}^{(M')}\in \Lambda_{\rm c}$ whose product produces $\OO$: $\OO=\mathcal{Q}^{(1)}\cdots \mathcal{Q}^{(M')}$. Here $f: \mathbb{R}_{\ge 0}\rightarrow \mathbb{R}_{\ge 0}$ is a monotonically increasing function.
\end{repdefn}

To understand the above definition, although each stabilizer $\OO$ may be expanded to checks in many ways (if the checks have redundancies), there is one (not necessarily unique) minimal expansion where the number of checks is smallest. Check soundness requires this minimal number to be bounded by a function of the weight of the original stabilizer.

We desire $f(M)$ to grow not too fast with $M$, and to not depend on $n$: this immediately rules out the unstable Ising version of toric code \eqref{eq:Ising_toric}, where $f(4)$ is already the order of the linear system size. Furthermore, this definition is closely related to local indistinguishability: \begin{enumerate}
    \item It is straightforward to show that LTO \eqref{eq:local_indisting} in finite dimensions guarantees check soundness for a power function $f$. We prove this claim in Proposition \ref{prop:finite_d}.
    \item Check soundness implies local indistinguishability in the following sense: A first direct implication is that, for any $S$ with $|S|$ sufficiently small, there exists a region $\bar{S}$ satisfying LTO \eqref{eq:local_indisting} with a bounded volume \begin{equation}
        |\bar{S}|< 4^{|S|} f(|S|). \label{eq:naiveconstraint}
    \end{equation}
    Here the $4^{|S|}$ factor is a very loose bound accounting for all possible Pauli strings in region $S$, while $f(|S|)$ assumes that each separate Pauli string might be ``check sound" on a different set of qubits.  Clearly, this bound is very loose, but it does imply that with check soundness, local operators ($|S|=\mathrm{O}(1)$) cannot differentiate the local ground states of a local part of $H_0$ in $\bar{S}$ near its support. Nevertheless, this simply-derived local indistinguishability condition, even without the $4^{|S|}$ factor, turns out to be insufficient for our proof of stability.  Importantly, check soundness gives us much tighter constraints than \eqref{eq:naiveconstraint}, because loosely speaking it also constrains the ``shape'' of $\bar{S}$; we will formulate this fact more precisely in Lemma \ref{lem:S<tfr}.
\end{enumerate} 

For CSS codes, one can look for check soundness separately for the $X$ part and $Z$ part: 
\begin{prop}\label{prop:CSS_sound}
    For a CSS code, if any $X$ stabilizer acting on $M< d_{\rm c}$ sites can be expanded as a product of at most $f(M)$ number of $X$ checks, and the same holds for $Z$ stabilizers, then the code has $(d_{\rm c},2f)$-check-soundness, where $(2f)(M):=2\cdot f(M)$.
\end{prop}
\begin{proof}
    For any stabilizer $\OO=\OO_X\cdot \OO_Z$ of the code acting on $M< d_{\rm c}$ sites, its $X$-part stabilizer $\OO_X$ acts on at most $M$ sites, so can be expanded to at most $f(M)$ $X$-checks. The same holds for the $Z$ part, so the whole stabilizer $\OO$ can be expanded to at most $f(M)+f(M)=2f(M)$ checks.
\end{proof}

\subsection{Finite-dimensional codes with local topological order}
We have introduced local topological order (LTO) in the main text; here we provide the formal definition for convenience: 
\begin{defn}[Local topological order]\label{def:LTO}
    For a stabilizer code that is local on a graph with vertex set $\Lambda$, we say it has LTO of distance $d_{\rm TO}$, if the following holds: For any ball $S_0\subset \Lambda$ with radius smaller than $d_{\rm TO}$, any operator $\OO$ supported on $S_0$ satisfies \begin{equation}\label{eq:local_indisting1}
    P_{\widetilde{S}} \OO P_{\widetilde{S}} = c_\OO P_{\widetilde{S}},
\end{equation}
for some number $c_\OO$, where $\widetilde{S}\supset S_0$ is a larger region that includes all checks touching \revise{$S_0$}.
\end{defn}

\revise{The maximum size $d_{\rm TO}$ of $S_0$ could scale with the system size.}
We clearly need $d_{\rm TO} < d$, as a logical operator $\OO$ would violate \eqref{eq:local_indisting1}. Nevertheless, notice that LTO implies check soundness with a power-law function $f$:

\begin{prop}\label{prop:finite_d}
    For a local stabilizer code in $\mathcal{D}$ spatial dimensions with code distance $d$, if it has LTO of distance $d_{\rm TO}<d$, then it has $(d_{\rm TO}, f)$-check-soundness, where \begin{equation}\label{eq:f_finiteD}
        f(M)=\Delta c_{\mathcal{D}} M^{\mathcal{D}}, 
    \end{equation}
    for some O(1) constant $c_{\mathcal{D}}$.
\end{prop}

\begin{proof}
For any stabilizer $\OO$ acting on $S$ with $|S|\le d_{\rm TO}$, decompose the support into its connected components: $S=S_1\cup\cdots\cup S_m$, such that each pair $(S_\mu,S_\nu)$ is disconnected on the graph $G$. The Pauli string $\OO=\OO_1\otimes \cdots\otimes \OO_m$ can be also decomposed. Since $|S_\mu|\le|S|< d$, each $\OO_\mu$ is not a logical. Moreover, $\OO_\mu$ commutes with all the checks, because otherwise the whole $\OO$ would fail to be a stabilizer. As a result, each $\OO_\mu$ is also a stabilizer. 

Since $S_\mu$ is connected on the $\mathcal{D}$-dimensional lattice, there exists a ball $B_\mu$ containing $S_\mu$ that has radius $|S_\mu|\le d_{\rm TO}$. Using the LTO condition, $\OO_\mu$ can be expanded as a product of checks inside $B_\mu$, whose number is at most $\Delta c_\mathcal{D}|S_\mu|^{\mathcal{D}}$ using \eqref{eq:Gamma<finite_d} and the LDPC condition.  This holds for all $\mu=1,\cdots,m$, so the stabilizer $\OO$ can be expanded to checks whose number is upper bounded by \begin{equation}
    \Delta c_\mathcal{D}\sum_{\mu=1}^m |S_\mu|^{\mathcal{D}} \le\Delta c_\mathcal{D}\lr{\sum_{\mu=1}^m |S_\mu|}^{\mathcal{D}} =\Delta c_{\mathcal{D}} |S|^{\mathcal{D}}.
\end{equation}
This establishes check soundness.
\end{proof}

Since most finite-dimensional codes, like the toric code, have LTO, they also have check soundness with a power-law function $f$, up to $d_{\rm c}$ that is of the order of the linear system size. As we will see, any power-law $f$ is sufficient for perturbative stability in finite dimensions.  As an explicit example, we illustrate in Figure \ref{fig:TC} that the two-dimensional toric code has quadratic check soundness $f(M)\sim M^2$, and this arises from the product of $M^2$ checks in a connected region reducing to a product of Pauli matrices along the perimeter of that region, which has length $\gtrsim M$.

\subsection{Constant-rate quantum codes}\label{app:constratequantum}
Quantum codes with a constant rate and a distance growing with $n$ have to live in ``infinite dimensions" \cite{BPT_bound10}, i.e. on expander graphs.  The most interesting codes satisfy the  scaling in \eqref{eq:Gamma<Delta}. Unfortunately, we could not prove stability for a power-law soundness function for these expanding codes; without further assumptions, we require it to grow slower than quadratic: \begin{equation}\label{eq:f_2_alpha}
    f(M)\le c_f M^{2-\beta},
\end{equation}
for positive constants $c_f, \beta$.

Intriguingly, \eqref{eq:f_2_alpha} is satisfied for all three known constructions of good quantum LDPC codes \cite{Panteleev_2022,qTanner_codes,dinur2022good}, where \emph{linear} check soundness holds with $\beta=1$, up to $d_{\rm c}=\Theta(n)$. Linear check soundness is the best scaling of soundness one can hope for in an LDPC code. Although the constructions for good codes are involved and we do not attempt to describe them in this paper, properties are proven for these codes that straightforwardly lead to linear check soundness. Since these codes are all CSS codes, it suffices to investigate the $X$ part and $Z$ part of the code separately due to Proposition \ref{prop:CSS_sound}.

The most obvious reference is \cite{dinur2022good}, where linear check soundness is just a special case of its Definition 2.13 for small-set (co)-boundary expansion, which is satisfied for the code by combining its Theorem 3.8 and Corollary 3.10. It is also pointed out there that linear check soundness leads to local testability \cite{dinur2022LTC} when viewing the $3$-term chain complex associated with the CSS code, as a classical code with redundant checks. 

As the constructions for good quantum codes in \cite{Panteleev_2022,qTanner_codes} also generate good local testable classical codes, these quantum codes are expected to also have linear check soundness.
Indeed, linear check soundness can be verified more directly:
For the good quantum code in \cite{Panteleev_2022}, it is shown in its proof of Theorem 2 that $d^{(1)}_{\rm LM}$ scales linearly with $n$, where the locally minimal distance $d^{(1)}_{\rm LM}$ is defined (in its Definition 5) as the size of the smallest stabilizer or logical $Z$-string $\OO$ that is locally minimal, i.e. multiplying any single $Z$-check on $\OO$ cannot decrease its support. The same holds for the locally minimal distance $d^{(1),*}_{\rm LM}$ for $X$-strings. The locally minimal distance lower bounds the code distance to guarantee the goodness of the code. Here we show that it also yields linear check soundness in a straightforward way:

\begin{prop}
    Let $d_{\rm c}=\min\lr{d^{(1)}_{\rm LM}, d^{(1),*}_{\rm LM}}$ be the smaller locally minimal distance of a CSS code. The CSS code has $(d_{\rm c}, f)$-check-soundness for $f(M)=M$.
\end{prop}

\begin{proof}
    For any $Z$ stabilizer $\OO$ with $M< d_{\rm c}$, since $M<d^{(1)}_{\rm LM}$, it is not locally minimal, so there exists a check that when adding to the stabilizer, decreases its weight by at least $1$. One can repeat this argument for the new stabilizer which is still not locally minimal, so that after at most $M$ steps the stabilizer shrinks to the identity operator. The original $\OO$ is thus a product of $\le M$ checks. The same holds for $X$ stabilizers, so the $(d_{\rm c}, f)$-check-soundness follows from Proposition \ref{prop:CSS_sound}.
\end{proof}

Therefore the quantum code in \cite{Panteleev_2022} has linear check soundness up to $d_{\rm c}=\Theta(n)$. The same holds for the quantum Tanner code constructed in \cite{qTanner_codes}; one can see this explicitly by Remark 8.9 in \cite{sum_square22}, which implies $d^{(1)}_{\rm LM}=\Theta(n)$ for the quantum Tanner code ($d^{(1),*}_{\rm LM}=\Theta(n)$ also follows because the code construction is symmetric against exchanging $X\leftrightarrow Z$ parts of the code).

Another notable family of LDPC CSS codes is the quantum expander codes \cite{qExpander15} based on the hypergraph product (HGP) \cite{HGP} of two classical expander codes\footnote{A classical expander code can be constructed by e.g. taking a random bipartite graph with regular degrees $(\Delta_{\rm B}, \Delta_{\rm C})$ and identifying the two sets of vertices as bits and checks. The parity-check matrix of the code is thus the adjacency matrix of the bipartite graph.}, which is the first quantum LDPC construction that achieves constant rate and a power-law distance $d=\sqrt{n}$. A minimalist definition of an HGP code is defined as follows:
\begin{defn}[Hypergraph product code]\label{def:HGP}
    Consider two classical codes where the bits for the first (second) code are labeled by $b$ ($\widetilde{b}$), and the checks for the first (second) code are labeled by $c$ ($\widetilde{c}$). Let $\partial,\widetilde{\partial}$ be defined for each code as follows: $\partial b$ is the set of checks $\widetilde{b}$ that act on $b$, while $\partial c$ is the set of bits contained in check $c$. The hypergraph product of these two classical codes is a quantum CSS code with qubits $\Lambda=\{b\widetilde{b}\}\cup \{c\widetilde{c}\}$ and checks $\Lambda_{\rm c}= \Lambda_{\rm c}^X \cup \Lambda_{\rm c}^Z$. Here each $Z$ check in $\Lambda_{\rm c}^Z$ is labeled by $b\widetilde{c}$, while each $X$ check in $\Lambda_{\mathrm{c}}^X$ is labeled by $c\widetilde{b}$, and each are defined as \begin{equation}
        \mathcal{Q}_{b\widetilde{c}} = \prod_{\widetilde{b}\in \widetilde{\partial} \widetilde{c}} Z_{b\widetilde{b}} \cdot \prod_{c\in\partial b} Z_{c\widetilde{c}}, \;\;\;\;\;\;\; \mathcal{Q}_{c\widetilde{b}} = \prod_{\widetilde{c}\in \widetilde{\partial} \widetilde{b}} X_{c\widetilde{c}} \cdot \prod_{b\in\partial c} X_{b\widetilde{b}}.
    \end{equation}
\end{defn}

See Fig.~\ref{fig:HGP}(a) for a sketch of the Tanner graph (Definition \ref{def:stabilizer_code}) of the quantum HGP code.
This relatively simple construction might enable experimental realizations \cite{LRESC_Yifan,HGP_Rydberg24,HGP_gate_Ryd24}, such as the prototype in \cite{Hong:2024fsg}. Intriguingly, the HGP structure naturally leads to check soundness \cite{singleshot_Campbell_2019}:
\begin{prop}[Lemma 5 in \cite{singleshot_Campbell_2019}]\label{prop:HGP_quad_sound}
    Any HGP code generated by \revise{taking the product of} a classical code with itself, has quadratic check soundness $f(M)=M^2/4$ up to $d_{\rm c}=\min(d_{\rm cla}, d_{\rm cla}^T)$, where $d_{\rm cla}, d_{\rm cla}^T$ are the distances of the classical code and its transpose code (by transposing the parity-check matrix).
\end{prop}

Although quadratic check soundness is not sufficient for our purpose, here we prove that quantum expander codes actually have linear check soundness up to $d_{\rm c}=\Theta(\sqrt{n})$. The idea is to exploit a closely-related concept: 
\begin{defn}[Check expansion]
    Let $\eta >0$. We say a stabilizer code has $(d_{\rm c}', \eta)$-linear-check-expansion, if for any set of $M'< d_{\rm c}'$ parity checks $\mathcal{Q}^{(1)},\cdots,\mathcal{Q}^{(M')}\in \Lambda_{\rm c}$, their product $\mathcal{Q}^{(1)}\cdots \mathcal{Q}^{(M')}$, as a stabilizer, acts on at least $\eta M'$ sites.
\end{defn}
This definition can be easily generalized beyond linear expansion functions like that for soundness; nevertheless this definition for linear case suffices for our purposes. Check expansion for classical codes can be defined in exactly the same way.  Notice that check expansion is weaker than check soundness for non-redundant codes: in that check expansion allows in principle for a product of $>d^\prime_{\mathrm{c}}$ checks to act on a small number of qubits.  For redundant codes, soundness allows for the product of many checks to be a small stabilizer as long as the product is not the smallest possible representation of the small stabilizer.

\begin{prop}\label{prop:HGP_linear_sound}
    Consider a CSS code on $n$ physical qubits that is the HGP of a classical code and itself, where the classical code has $n_{\rm cla}$ physical bits and $m_{\rm cla}=\zeta n_{\rm cla}$ checks. Suppose the classical code, together with its transpose code, have linear distance $d_{\rm cla}, d_{\rm cla}^T\ge \xi n_{\rm cla}$, and has $(\chi n_{\rm cla}, \eta)$-linear-check-expansion. Then the HGP quantum code has $(\chi' \sqrt{n}, f)$-check-soundness with $f(M)=6M/\eta$, and \begin{equation}\label{eq:chi'=chi}
        \chi'=\frac{1}{\sqrt{1+\zeta^2}}\min(\sqrt{2}\chi, \xi).
    \end{equation}
\end{prop}

This establishes linear check soundness for quantum expander codes with $d_{\rm c}=\Theta(\sqrt{n})$, so long as its mother classical expander code satisfies all the conditions in Proposition \ref{prop:HGP_linear_sound} with $\zeta,\xi,\chi,\eta$ being constants. More precisely, the linear-check-expansion property comes from the unique neighbor expansion property for a random bipartite graph of degree $(\Delta_{\rm B}, \Delta_{\rm C})$, where $\Delta_{\rm B},\Delta_{\rm C}>2$; see e.g. Theorem A.6 in \cite{Hong:2024vlr}, or \cite{ModernCodingTheory}, for details. 

\begin{figure}
    \centering
    \includegraphics[width=0.9\linewidth]{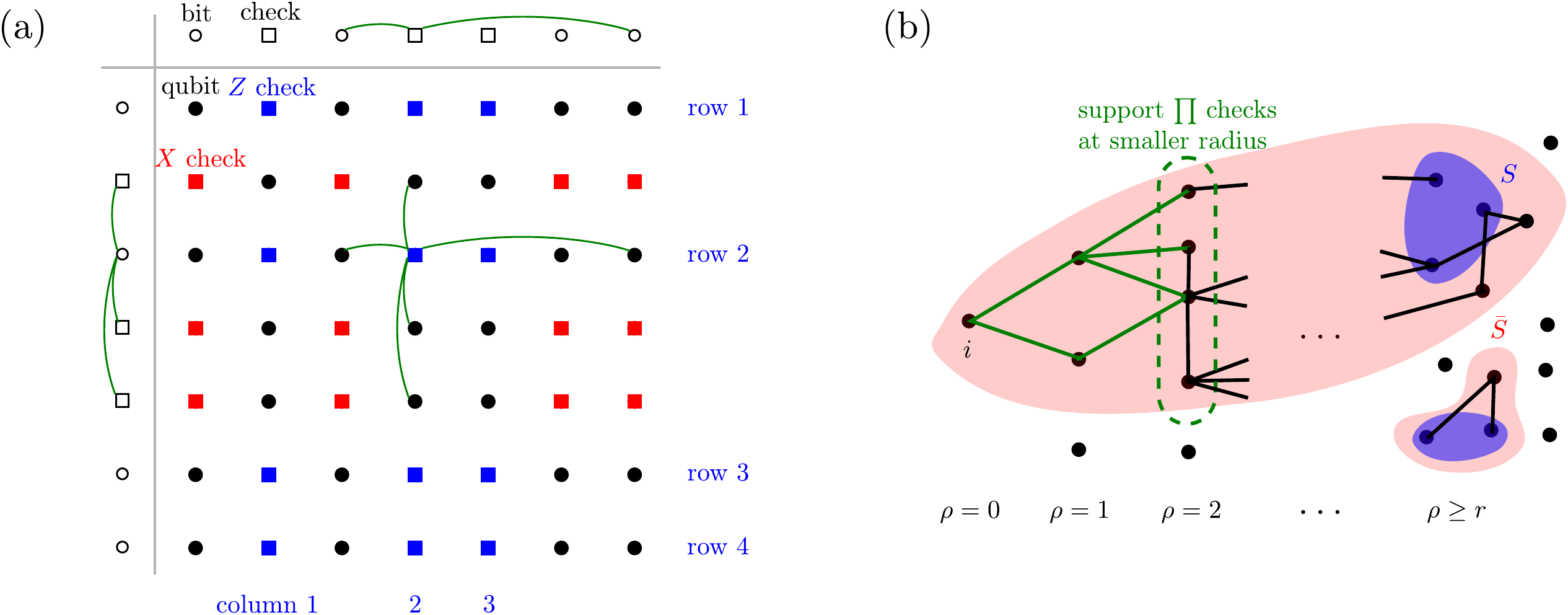}
    \caption{(a) An illustration of the Tanner graph of a quantum hypergraph product code obtained from two classical codes. Linear check soundness of the quantum code can be obtained from expansion properties of the mother classical codes, as we show in Proposition \ref{prop:HGP_linear_sound}. (b) A sketch of obtaining local indistinguishability \eqref{eq:local_indisting} from check soundness. A stabilizer in $S$ (which can be disconnected as shown by the shaded blue regions) is minimally expanded into product of checks in $\bar{S}$ (shaded red regions), where the checks are shown by edges connecting the qubits (black dots). Here we plot each check as acting on $2$ qubits for simplicity. For $\bar{S}$ to contain $i$ that is distance $r$ away from $S$, the checks need to ``push'' its support outwards to larger radius while constrained the check soundness condition. This results in a $r$-determined lower bound on $|S|$. }
    \label{fig:HGP}
\end{figure}

\begin{proof}[Proof of Proposition \ref{prop:HGP_linear_sound}]

We first show that the soundness result of Proposition \ref{prop:HGP_linear_sound} follows if the HGP code has $((\chi')^2n/4, \eta/6)$-linear-check-expansion: Any stabilizer of weight \begin{equation}\label{eq:stabilizer<distance}
    M<\chi' \sqrt{n}\le \xi n_{\rm cla}\le \min(d_{\rm cla}, d_{\rm cla}^T),
\end{equation}
is expandable to $M'\le M^2/4<(\chi')^2n/4$ checks due to Proposition \ref{prop:HGP_quad_sound}. Here in \eqref{eq:stabilizer<distance} we have used \eqref{eq:chi'=chi}, and \begin{equation}\label{eq:n<ncla}
    n=n_{\rm cla}^2 + m_{\rm cla}^2\le (1+\zeta^2)n_{\rm cla}^2,
\end{equation}
from Definition \ref{def:HGP} of HGP. Because the number of checks $M'$ is bounded by the linear-check-expansion cutoff, we can apply the check-expansion property to these $M'$ checks and get $M\ge \eta M'/6$, so that $M'\le 6M/\eta$.

We then use the check expansion of the classical code to prove the linear-check-expansion property above of the HGP code.
We first focus on $Z$-checks, and aim to show that the $Z$-checks, as checks of a classical code, have ($\chi_0,\eta_0$)-linear-check-expansion with \begin{equation}\label{eq:Zcheck_exp}
 (\chi_0,\eta_0) =    (\chi{^\prime 2} n/4, \eta/3).
\end{equation}
The same property will hold for $X$-checks as well because of the $X\leftrightarrow Z$ symmetry of the construction. Since the $Z$ and $X$ checks are independent, this will lead to the desired $((\chi')^2n/4, \eta/6)$-linear-check-expansion for the quantum HGP code, analogously to Proposition \ref{prop:CSS_sound}.

Observe that the ($Z$-)checks are organized in rows and columns, as shown in Fig.~\ref{fig:HGP}(a). Consider a given set of $M<(\chi')^2n/4$ ($Z$-)checks $\mathsf{Q}\subset \Lambda_{\rm c}$ that multiply to stabilizer $\OO$. For each row labeled by $\mathsf{r}$, if it contains \begin{equation}\label{eq:Mr<ncla}
    M_{\mathsf{r}}<\chi n_{\rm cla}
\end{equation} 
checks in $\mathsf{Q}$, then by the check expansion of the classical code, there are at least $\eta M_{\mathsf{r}}$ qubits in this row that $\OO$ acts nontrivially on. Here we have used the fact that qubits in this row are not acted on by ($Z$-)checks in other rows. We say row $\mathsf{r}$ is unfull (and full otherwise) if \eqref{eq:Mr<ncla} is satisfied (violated). The same holds for each column $\mathsf{c}$: \begin{equation}
    M_{\mathsf{c}}<\chi n_{\rm cla}, \rarrow \OO \text{ acts on } \ge \eta M_{\mathsf{c}} \text{ qubits in column } \mathsf{c}. \label{eq:unfullcolumn0}
\end{equation}

If the set of all unfull rows contains a sufficient portion $\ge 1/3$ of the ``weight'' $M$ of the checks in $\mathsf{Q}$: \begin{equation}\label{eq:unfull_row<}
    \sum_{\mathsf{r} \text{ unfull}}M_{\mathsf{r}} \ge \frac{M}{3},
\end{equation}
then the stabilizer $\OO$ acts on at least $\eta M/3$ qubits from the classical check expansion property for these rows, so that the desired check-expansion \eqref{eq:Zcheck_exp} is satisfied for this particular set of $M$ checks.

Otherwise if \eqref{eq:unfull_row<} does not hold, the set of \emph{full} rows must contain a sufficient portion of the check weight: \begin{equation}\label{eq:full_row>}
    \sum_{\mathsf{r} \text{ full}}M_{\mathsf{r}} \ge \frac{2M}{3}.
\end{equation}
For this case, we aim to show that unfull columns contain sufficient check weight: \begin{equation}\label{eq:unfull_column}
    \sum_{\mathsf{c} \text{ unfull}}M_{\mathsf{c}} \ge \frac{M}{3},
\end{equation}
so that the stabilizer still acts on at least $\eta M/3$ qubits from the argument above with rows and columns switched. To this end, we separate the checks in $\mathsf{Q}=\mathsf{Q}_{\rm full}\cup \mathsf{Q}_{\rm unfull}$ into those contained in the full rows, and others in the unfull rows, with \begin{equation}\label{eq:Q_unfull<}
    |\mathsf{Q}_{\rm unfull}|< \frac{M}{3},
\end{equation}
from \eqref{eq:full_row>}. The number of full rows is bounded by \begin{equation}\label{eq:full_row<}
    |\{\mathsf{r}: \mathsf{r} \text{ full}\}|\le \frac{M}{\chi n_{\rm cla}} < \frac{(\chi')^2n}{4\chi n_{\rm cla}}\le \frac{1}{2}\chi n_{\rm cla},
\end{equation}
using \eqref{eq:chi'=chi} and \eqref{eq:n<ncla}.

If there are no other checks $|\mathsf{Q}_{\rm unfull}|=0$, then because of \eqref{eq:full_row<} and \eqref{eq:unfullcolumn0}, all columns are unfull, which implies that \eqref{eq:unfull_column} holds.

We then think of adding back the checks in $\mathsf{Q}_{\rm unfull}$ adversarially to make some of the columns full. \eqref{eq:full_row<} implies that each time a column is made full, it consists of $\ge \chi n_{\rm cla}/2$ checks in $\mathsf{Q}_{\rm unfull}$ and $< \chi n_{\rm cla}/2$ checks in $\mathsf{Q}_{\rm full}$. In other words, by costing $x$ checks in $\mathsf{Q}_{\rm unfull}$, it only decreases the quantity in \eqref{eq:unfull_column} by at most $x$. As a result, costing all checks in $\mathsf{Q}_{\rm unfull}$ only decreases the quantity in \eqref{eq:unfull_column} by an amount bounded by \eqref{eq:Q_unfull<}. Therefore, \begin{equation}
    \sum_{\mathsf{c} \text{ unfull}}M_{\mathsf{c}} \ge |\mathsf{Q}_{\rm full}|-|\mathsf{Q}_{\rm unfull}|\ge M/3,
\end{equation}
from \eqref{eq:Q_unfull<},
which proves \eqref{eq:unfull_column}. Recall that this establishes linear-check-expansion \eqref{eq:Zcheck_exp} for the $Z$ and $X$ checks separately, which combines to $((\chi')^2n/4, \eta/6)$-linear-check-expansion for the whole quantum code. Proposition \ref{prop:HGP_linear_sound} then holds by the argument at the beginning of the proof, which boosts the quadratic check soundness of Proposition~\ref{prop:HGP_quad_sound} to linear.
\end{proof}

\subsection{Local indistinguishability from check soundness}
We now address the main  concern of this section -- showing that check soundness provides strong constraints on locality for stabilizer operators.  We will see that this technical issue is crucial in the proof of our main theorem.

Let $S\subset \Lambda$ be a set of distance $r=\mathsf{d}(i,S)$ away from a given site $i\in \Lambda$. Suppose there exists a stabilizer $\OO$ in $S$ such that its minimal expansion contains a check acting on $i$. Given check soundness and $|S|\le d_{\rm c}$, how large should $|S|$ be for this to happen? 

First of all, $|S|$ cannot be a constant if $r$ is large: Any stabilizer in $S$ can be minimally-expanded to at most $f(|S|)$ number of checks that are all supported in a region $\bar{S}$. The extra support $\bar{S}\setminus S$ has to be connected to $S$, because the expansion is not minimal otherwise. Therefore, in order for $\bar{S}$ to contain $i$ as desired, we need it to connect $S$ and $i$ such that \begin{equation}\label{eq:r<barS}
    r\le |\bar{S}|\le \Delta\cdot f(|S|),
\end{equation}
which gives a lower bound on $|S|$ by inverting function $f$. For linear check soundness, $|S|=\Omega(r)$ for example.
Here the second inequality in \eqref{eq:r<barS} is because the $\le f(|S|)$ checks in the minimal expansion acts on at most $\Delta f(|S|)$ sites. 

The above argument can be ``bootstrapped'' to yield a much stronger bound than \eqref{eq:r<barS}, for which we sketch the idea in Fig.~\ref{fig:HGP}(b): Suppose the minimal expansion of stabilizer $\OO$ has the form \begin{equation}\label{eq:O=Orho}
    \OO=\prod_{\rho=1}^{r+1} \OO_{\rho}, 
\end{equation}
where $\OO_\rho$ for $\rho\le r$ is product of $g(\rho)$ checks that are contained in ball $B_{i,\rho}$ but not $B_{i,\rho-1}$, and $\OO_{r+1}$ is product of the remaining checks not contained in $B_{i,r}$. 

If $g(1)\ge 1$ checks act on $i$, then the product of these checks $\OO_1$ is supported in $B_{i,1}\setminus B_{i,0}$, on at least $f^{-1}(g(1))$ sites. The reason is that these checks is a minimal expansion of $\OO_1$: Otherwise \eqref{eq:O=Orho} is not a minimal expansion for $\OO$. The weight of the stabilizer $\OO_1$ is then no smaller than $f^{-1}(g(1))$ due to check soundness, if $f^{-1}(g(1))<d_{\rm c}$. As a consequence, $\Delta\cdot g(2)\ge f^{-1}(g(1))$ because the $g(2)$ checks in the next layer has to move the support of the product outward (i.e. stabilizer $\OO_1\OO_2$ is supported in $B_{i,2}\setminus B_{i,1}$), while each check acts on at most $\Delta$ sites.

The above argument holds for any layer, so we have iteration relations \begin{equation}
    \Delta\cdot g(\rho+1) \ge f^{-1}(g(1)+\cdots + g(\rho)), \quad \mathrm{if}\quad f^{-1}(g(1)+\cdots + g(\rho))<d_{\rm c}.
\end{equation}
This is solved by $g(1)+\cdots + g(\rho) \ge \widetilde{g}(\rho)$, where the increasing function $\widetilde{g}(\rho)$ is determined by the following iteration relation \begin{equation}\label{eq:g=g+f}
    \widetilde{g}(\rho+1) = \widetilde{g}(\rho) + \min\mlr{d_{\rm c}, f^{-1}(\widetilde{g}(\rho)) }/\Delta, \where \widetilde{g}(1)= 1.
\end{equation}
Here the minimum function comes from the fact that, if $f^{-1}(\widetilde{g}(\rho))\ge d_{\rm c}$, the next layer needs at least $d_{\rm c}$ support to satisfy check soundness.

Since $S$ is minimally expanded with $|S|< d_{\rm c}$, check soundness guarantees that \begin{equation}
    f(|S|) \ge g(1)+\cdots + g(r) \ge \widetilde{g}(r), \rarrow |S|\ge \widetilde{f}(r),
\end{equation}
where the increasing function $\widetilde{f}(r)= f^{-1}(\widetilde{g}(r))$ is a function determined by $\widetilde{g},f$, and thus determined by $f$ using the following iteration relation obtained from \eqref{eq:g=g+f}: \begin{equation}\label{eq:tf=f}
    \widetilde{f}(r+1) := f^{-1}\glr{ f\lr{\widetilde{f}(r)} + \min\mlr{d_{\rm c}, \widetilde{f}(r) }/\Delta}, \where \widetilde{f}(1) = f^{-1}(1).
\end{equation}
We summarize the above by the following Lemma: 
\begin{lem}\label{lem:S<tfr}
    For a set $S$, if $|S|< \min(d, d_{\rm c})$ and $|S|<\widetilde{f}(r)$ for some $r$ with function $\widetilde{f}$ determined by $f$ in \eqref{eq:tf=f}, then any stabilizer in $S$ can be expanded into product of checks supported in region $\bar{S}=\{i\in\Lambda: \mathsf{d}(i,S)\le r\}$. Furthermore, the enlarged region $\bar{S}$ has local indistinguishability in the sense that \eqref{eq:local_indisting} holds for any $\OO$ in $S$.
\end{lem}

\begin{proof}
    We have shown the stabilizer property above: There we focused on a given site $i$, but we have actually shown that the minimal expansion of any stabilizer in $S$ cannot act on any site $i$ that is distance $r$ away, given $|S|< d_{\rm c}$ and $|S|<\widetilde{f}(r)$. 
    
    The Lemma then follows from this stabilizer property, using an argument closely following Lemma 1 in \cite{topo_Hastings}
    : Since \eqref{eq:local_indisting} is linear in $\OO$, it suffices to investigate the three classes of Pauli strings: \begin{enumerate}
        \item $S$ does not contain logical operators $\OO$ due to $|S|<d$. 
        \item Any error operator $\OO$ in $S$ anti-commutes with at least one check in $S$ and thus has $P_S \OO P_S=0$.
        \item Any stabilizer $\OO$ in $S$ is expanded into checks contained in $\bar{S}$, so acts as identity in the subspace $P_{\bar{S}}$.
    \end{enumerate} 
    A general operator $\OO$ in $S$ thus satisfies \eqref{eq:local_indisting} by expanding to Paulis. Note that the number $c_\OO$ is only contributed by stabilizers, which have syndrome $\bs=\emptyset$.
\end{proof}

If $f$ grows slowly, function $\widetilde{f}(r)$ will increase with $r$ fast enough. We will need this property in the stability proof to fight against the possible fast expansion of the graph. This is given precisely by the following Lemma:

\begin{lem}\label{lem:f_growth}
If either of the following holds: \begin{enumerate}
    \item $f$ grows slower than quadratic as \eqref{eq:f_2_alpha} for $\beta>0$.
    \item $H_0$ has finite dimension $\mathcal{D}$ and $f$ grows slower than a polynomial of bounded degree: \eqref{eq:f_2_alpha} holds for a possibly negative constant $\beta$.
\end{enumerate}
Then there exist constants $c_f',c_f'', \alpha>0$ determined by $\Delta, c_f,\beta$ (and $\mathcal{D},c_{\mathcal{D}}$) such that 
\begin{equation}\label{eq:sum_r<}
        \sum_{r\ge 0:\, \widetilde{f}(r)<d_{\rm c}}\gamma(r) \ee^{-\delta\kappa \widetilde{f}(r)} \le c_f''\, \ee^{c_f'(\delta\kappa)^{-\alpha}},\quad \forall\, \delta\kappa>0,
    \end{equation}
    where $\widetilde{f}$ is the function determined by $f$ in \eqref{eq:tf=f}.
\end{lem}
Here $\delta\kappa$ is a single variable \revise{that will appear} in the context of the proof of Theorem \ref{thm:qtm}. \revise{In a nutshell, as we do Schrieffer-Wolff transformations the operators in the Hamiltonian become more and more nonlocal at higher orders.  To obtain bounds of the form $\lVert [A,B]\rVert_{\kappa-\delta \kappa} \lesssim \lVert A\rVert_\kappa \lVert B\rVert_\kappa$ requires reducing the value of $\kappa$ by a small amount at each order in perturbation theory: see Lemma \ref{lem:clustexp}.}

\begin{proof}
We assume $\beta <1$ without loss of generality, because the linear check soundness $\beta =1$ case also satisfies \eqref{eq:f_2_alpha} with $\beta=0.9$, for example. 

Plugging \eqref{eq:f_2_alpha} into the iteration for $\widetilde{g}(r)$ in \eqref{eq:g=g+f}, we have \begin{equation}\label{eq:g>g+f}
    \widetilde{g}(r+1) - \widetilde{g}(r)\ge \frac{1}{\Delta} \lr{\frac{\widetilde{g}(r)}{c_f}}^{\frac{1}{2-\beta}}.
\end{equation}
Here we have focused on regime $r<r_{\rm c}$ where $r_{\rm c}$ is the smallest integer $r$ such that $f^{-1}(\widetilde{g}(r))=\widetilde{f}(r)\ge d_{\rm c}$. This is the only regime that contributes to the sum \eqref{eq:sum_r<} that we want to bound, and we have taken $f^{-1}(\widetilde{g}(r))$ in \eqref{eq:g=g+f} to be the minimum to get \eqref{eq:g>g+f}. 
\eqref{eq:g>g+f} is solved by ansatz $\widetilde{g}(r) = c_g r^{\frac{2-\beta}{1-\beta}}$, because its left hand side becomes \begin{equation}
    \widetilde{g}(r+1) - \widetilde{g}(r) \ge c_g \left.\frac{\dd \widetilde{g}}{\dd r}\right|_{r} = \frac{2-\beta}{1-\beta}c_g  r^{\frac{1}{1-\beta}} = \frac{2-\beta}{1-\beta}c_g^{1-\frac{1}{2-\beta}} \mlr{\widetilde{g}(r)}^{\frac{1}{2-\beta}},
\end{equation}
(the inequality is because the function $\frac{\dd \widetilde{g}}{\dd r}$ is monotonically increasing) which equals the right hand side by choosing $c_g$ to be given by $\frac{2-\beta}{1-\beta}c_g^{1-\frac{1}{2-\beta}}=c_f^{-\frac{1}{2-\beta}}/\Delta$. As a result, \begin{equation}\label{eq:tfr>}
    \widetilde{f}(r)=f^{-1}\lr{\widetilde{g}(r)} \ge \widetilde{c}_f r^{\frac{1}{1-\beta}},
\end{equation}
for some constant $\widetilde{c}_f$ determined by $\Delta,c_f$ and $\beta$. 

We first deal with the general graph case. We use \eqref{eq:Gamma<Delta} and \eqref{eq:tfr>} to get \begin{align}\label{eq:beta2alpha}
    \sum_{r\ge 0:\, \widetilde{f}(r)<d_{\rm c}}\gamma(r) \ee^{-\delta\kappa \widetilde{f}(r)} &\le \sum_{r\ge 0}\Delta^r \ee^{-\widetilde{c}_f\delta\kappa  r^{\frac{1}{1-\beta}} } =\sum_{r\ge 0} \ee^{r\log \Delta -\widetilde{c}_f\delta\kappa  r^{\frac{1}{1-\beta}} } \nonumber\\
    &\le \lr{2r_{\rm max}+\frac{1}{\widetilde{c}_f'}} \max_{r\ge 0} \ee^{r\log \Delta -\widetilde{c}_f\delta\kappa  r^{\frac{1}{1-\beta}} } = \lr{2r_{\rm max}+\frac{1}{\widetilde{c}_f'}} \ee^{\beta r_{\rm max}\log \Delta} \nonumber\\ 
    &\le c_f'' \ee^{2\beta r_{\rm max}\log \Delta} = c_f'' \exp\mlr{\beta (\log \Delta) \lr{\frac{(1-\beta)\log\Delta}{\widetilde{c}_f\delta\kappa}}^{\frac{1-\beta}{\beta}} },
\end{align}
Here to get the second line, we used the fact that function $\ee^{r\log \Delta -\widetilde{c}_f\delta\kappa  r^{\frac{1}{1-\beta}} }$ decays faster than an exponential $\propto (1-\widetilde{c}_f')^r$ for $r>2r_{\rm max}$, where $r_{\rm max}=\lr{\frac{(1-\beta)\log\Delta}{\widetilde{c}_f\delta\kappa}}^{\frac{1-\beta}{\beta}}$ maximizes the function, and constant $\widetilde{c}_f'$ does not depend on $\delta\kappa$. As a result, the sum over $r$ is upper bounded by the maximum times $2r_{\rm max}+\frac{1}{\widetilde{c}_f'}$, where the second term accounts for the exponentially decaying contribution $\propto 1+(1-\widetilde{c}_f') + (1-\widetilde{c}_f')^2 + \cdots=\frac{1}{\widetilde{c}_f'}$ from $r>2r_{\rm max}$. To get the last line,
we used $2r_{\rm max} \le \lr{c_f''- \frac{1}{\widetilde{c}_f'}}\ee^{\beta r_{\rm max}\log \Delta}$ ($\forall r_{\rm max}\ge 0$) for a constant $c_f''$.  \eqref{eq:beta2alpha} is then established by \eqref{eq:sum_r<} for $\alpha = \frac{1-\beta}{\beta}$ and $c_f'=2\beta (\log\Delta)\lr{\frac{(1-\beta)\log\Delta}{\widetilde{c}_f }}^{\frac{1-\beta}{\beta}}$.

We then treat finite dimensions where $\beta$ can be negative. \eqref{eq:Gamma<finite_d} and \eqref{eq:tfr>} lead to
\begin{align}\label{eq:beta2alpha1}
    \sum_{r\ge 0:\, \widetilde{f}(r)<d_{\rm c}}\gamma(r) \ee^{-\delta\kappa \widetilde{f}(r)} &\le \sum_{r\ge 0}c_{\mathcal{D}}r^{\mathcal{D}-1} \ee^{-\widetilde{c}_f\delta\kappa  r^{\frac{1}{1-\beta}} } \le c_{\mathcal{D}}\lr{2r_{\rm max}+\frac{1}{\widetilde{c}_f'}} \max_{r\ge 0} r^{\mathcal{D}-1} \ee^{-\widetilde{c}_f\delta\kappa  r^{\frac{1}{1-\beta}} } \nonumber\\ 
    &\le c_{\mathcal{D}}\lr{2r_{\rm max}+\frac{1}{\widetilde{c}_f'}} r_{\rm max}^{\mathcal{D}-1} \le c_f'' \ee^{r_{\rm max}} = c_f'' \exp\mlr{\lr{\frac{(1-\beta)(\mathcal{D}-1)}{\widetilde{c}_f\delta\kappa}}^{1-\beta} },
\end{align}
Here the function $r^{\mathcal{D}-1} \ee^{-\widetilde{c}_f\delta\kappa  r^{\frac{1}{1-\beta}} }$ is analyzed in the same way as \eqref{eq:beta2alpha}, which decays exponentially (with exponent $\widetilde{c}_f'$ being a constant) for $r>2r_{\rm max}$ with $r_{\rm max}=\lr{\frac{(1-\beta)(\mathcal{D}-1)}{\widetilde{c}_f\delta\kappa}}^{1-\beta}$. We then bound the power function of $r_{\rm max}$ by an exponential with a constant prefactor $c_f''$. \eqref{eq:sum_r<} is then also satisfied for $\alpha=1-\beta$ and $c_f'=\lr{\frac{(1-\beta)(\mathcal{D}-1)}{\widetilde{c}_f}}^{1-\beta}$.
\end{proof}

\section{Stability of quantum codes}\label{sec:main_proof}
We are now ready to state our main theorem:

\begin{repthm}{thm:qtm}[Formal version]
    Let $H_0$ be a stabilizer code Hamiltonian of the form \eqref{eq:H0=} that has (\emph{i}) LDPC property with bounded degree $\Delta$, (\emph{ii}) a gap $\ge1$ from \eqref{eq:minlambdacheck}, (\emph{iii}) code parameters $\llbracket n,k,d\rrbracket$, and (\emph{iv}) $(d_{\rm c}, f)$-check-soundness with \begin{equation}\label{eq:d>logn}
        d_{\rm s}:= \min(d, d_{\rm c}) \ge c_d \log n,
    \end{equation}
    where $c_d>0$ is some fixed constant independent of $n$, and $f$ grows slowly such that \eqref{eq:sum_r<} holds 
    for constants $c_f',c_f'',  \alpha>0$. Let $\kappa_1$ be a constant satisfying \begin{equation}\label{eq:kappa1>cd}
        \kappa_1 > 2/c_d,
    \end{equation}
    There exist positive constants $\epsilon_0,c_1,c_2,c_3$ determined by $\Delta,\kappa_1,c_f',c_f'',\alpha$, such that if $n\ge c_3$, for any perturbation $V=\sum_{S,\bs} V_{S,\bs}$ with \begin{equation}\label{eq:V=eps}
        \norm{V}_{\kappa_1} =:\epsilon \le \min\lr{\epsilon_0,\frac{1}{3c_1}},
    \end{equation}
    the spectrum of $H:=H_0+V$, up to an overall energy shift, is contained in two disjoint intervals \begin{equation}\label{eq:spec_H}
        [-\epsilon_*,\epsilon_*] \cup[1-c_1\epsilon-\epsilon_*, \infty)\quad \subset \quad [-\epsilon_*,\epsilon_*] \cup[\frac{1}{2}+\epsilon_*, \infty),
    \end{equation} 
    where there are $2^k$ eigenstates in the first interval, and \begin{equation}\label{eq:eps*=}
        \epsilon_*= c_2 n\epsilon\, \ee^{-\kappa_1 d_{\rm s}/2 },
    \end{equation}
    decays to zero as $n\rightarrow \infty$ due to \eqref{eq:d>logn} and \eqref{eq:kappa1>cd}. 
    In other words, there remains a gap $\ge 1/2$ separating the $2^k$ ground states from the rest of the spectrum, and the ground states are almost degenerate with small energy difference $\le 2\epsilon_*$. Moreover, the projector $P_{\rm new}$ to the perturbed ground state subspace is related to the unperturbed $P$ by \begin{equation}\label{eq:Pnew-P}
        \norm{P_{\rm new}-U P U^\dagger} \le 4\sqrt{\epsilon_*},
    \end{equation}
    where \begin{equation}\label{eq:U=intA}
        U=\mathcal{T}\exp\mlr{\int^1_0 \dd t A(t)},
    \end{equation}
    ($\mathcal{T}$ means time-ordering) is finite-time evolution generated by a time-dependent anti-Hermitian $A(t)$ with local norm \begin{equation}\label{eq:At<}
        \max_t \norm{A(t)}_{\kappa_1/2} \le 2 \epsilon.
    \end{equation}
\end{repthm}

Our claim that known instances of good quantum LDPC codes, together with the quantum HGP expander codes, form robust phases of matter with finite entropy density at zero temperature, immediately follows from combining Theorem \ref{thm:qtm} with Lemma \ref{lem:f_growth}, along with the results catalogued in Appendix \ref{app:constratequantum}.

\subsection{Schrieffer-Wolff transformations}

We have reviewed the Schrieffer-Wolff transformation (SWT) scheme for the proof in the main text; here we introduce some notations \revise{and} explicit calculations on performing one SWT step, which will be useful in the proof:

For any operator $V=\sum_{S,\bs} V_{S,\bs}$, define $\PP V:=\sum_{S,\bs} (\PP V)_{S,\bs}$ where \begin{equation}
    (\PP V)_{S,\bs} := P_S V_{S,\bs} P_S+ Q_S V_{S,\bs} Q_S.
\end{equation}
$(\PP V)_{S,\bs}$ indeed has syndrome $\bs$ because $P_S,Q_S$ commutes with all checks. As a result, since $V_{S,\bs}$ is strongly supported in $S$, $(\PP V)_{S,\bs}$ is also strongly supported in $S$. We will also write $(\PP V)_{S,\bs}$ as $\PP V_{S,\bs}$, although strictly speaking, $\PP$ is not a global superoperator that maps operators to operators, because it treats different terms with support $S$ in a different way. $\PP V$ is locally block diagonal in the following sense: \begin{defn}\label{def:D'}
We say $D$ is locally block-diagonal, if its decomposition $D=\sum_{S,\bs} D_{S,\bs}$ satisfies \begin{equation}\label{eq:DcommutP}
    [D_{S,\bs}, P_S ]=0,\quad \forall S,\bs.
\end{equation}
\end{defn}

On the other hand, $V=\PP V + \PP^\perp V$ where each local term \begin{equation}
    \lr{\PP^\perp V}_{S,\bs} := P_S V_{S,\bs} Q_S+ Q_S V_{S,\bs} P_S,
\end{equation}
is block-off-diagonal. Both $\PP$ and $\PP^\perp$ cannot increase operator norm: \begin{equation}\label{eq:PP<1}
    \norm{\PP V_{S,\bs}}, \norm{\PP^\perp V_{S,\bs}}\le \norm{V_{S,\bs}}, \rarrow \norm{\PP V}_\kappa, \norm{\PP^\perp V}_\kappa \le \norm{V}_\kappa,
\end{equation}
because, for example, \begin{equation}
    \norm{\PP V_{S,\bs}}\le \max\lr{\norm{P_S V_{S,\bs} P_S}, \norm{Q_S V_{S,\bs} Q_S}} \le \max\lr{\norm{V_{S,\bs} }, \norm{V_{S,\bs}}} =\norm{V_{S,\bs} },
\end{equation}
where the maximization comes from the orthogonality of the subspaces that $P_S$ and $Q_S$ map to.

The following equation \begin{equation}\label{eq:H0A+V=PV}
    [H_0, A] + V = \PP V,
\end{equation}
is solvable by anti-Hermitian $A=\sum_{S,\bs} A_{S,\bs}$ with \begin{equation}\label{eq:A=PVQH}
    A_{S,\bs} = P_S V_{S,\bs} Q_S H^{-1}_S - H^{-1}_SQ_S V_{S,\bs} P_S,
\end{equation}
as one can verify explicitly. Note that the (pseudo)inverse of $H_S$ above is well-defined because it is always multiplied by $Q_S$ that projects out the kernel.

\subsection{Proof of Theorem \ref{thm:qtm}}

\subsubsection{Overview}

\revise{As already sketched in the main text,} our ultimate goal is to construct a unitary \eqref{eq:U=intA} that approximately block-diagonalizes (in terms of the two subspaces $P$ and $I-P$) the Hamiltonian: \begin{equation}\label{eq:UHU=star}
    U^\dagger H U=H_0+D_*+V_*+c_* I,
\end{equation}
where $c_*$ is a number, and we desire the following properties:
\begin{enumerate}
    \item The ``remaining perturbation'' $V_*$ is (nearly) exponentially small in distance:
    \begin{equation}\label{eq:V*<delta}
        \norm{V_*} \le \epsilon_*,
    \end{equation}
    with $\epsilon_*$ given by \eqref{eq:eps*=}.
    \item $D_*$ is relatively bounded by $H_0$: \begin{equation}\label{eq:D*<H0*}
        \norm{D_*\ket{\psi}}\le c_1 \epsilon \norm{H_0\ket{\psi}},
    \end{equation}
    for some constant $c_1$ determined by $\Delta,\kappa_1$. \eqref{eq:D*<H0*} is equivalent to the operator inequality $D_*^2 \le (c_1\epsilon)^2H_0^2$, and implies $D_*$ is block-diagonal because $D_*\ket{\mathrm{GS}}=0$ for any ground state $\ket{\mathrm{GS}}$ of $H_0$.
    \item The generator $A(t)$ of $U$ satisfies the locality bound \eqref{eq:At<}.
\end{enumerate} 

We aim to show these properties, because the theorem follows by the following Lemma that we will prove in Appendix \ref{sec:spec}:
\begin{lem}\label{lem:spec}
    If there exists a unitary $U$ such that \eqref{eq:UHU=star} holds with the three properties above satisfied, then Theorem \ref{thm:qtm} holds.
\end{lem}

Our proof then involves three parts: \begin{enumerate}
    \item We first construct $U$ by a perturbation series of iterated SWTs, truncated at a high order $m_*\sim d_{\rm s}$.
    \item We show that the perturbation series converges by obtaining strong locality bounds (Lemma \ref{lem:iterate}) on the operators involved in the SWTs.
    \item Finally, we use the locality bounds to verify the above properties one by one, so that the proof concludes by invoking Lemma \ref{lem:spec}. 
\end{enumerate}

\subsubsection{Setting up the Schrieffer-Wolff transformation}
Recall the strategy of SWTs introduced in Section \ref{sec:main}.  Suppose at the $m$-th step, we have constructed a unitary $U_{m-1}$ such that the Hamiltonian is rotated to \begin{equation}\label{eq:UkHUk=}
    U^\dagger_{m-1} H U_{m-1}=H_0+D_m+V_m+E_m.
\end{equation}
Here $D_m$ is the locally-block-diagonal effective Hamiltonian at this order that preserves subspace $P$, and $V_m$ is the remaining perturbation that is roughly suppressed to order $\epsilon^m$. Both of them contain only terms with support smaller than the code distance: \begin{equation}
    D_m = \sum_{S:|S|<d_{\rm s}} D_{m,S}, \quad V_m = \sum_{S:|S|<d_{\rm s}} V_{m,S}.
\end{equation}
All terms of larger support are contained in the ``garbage'' $E_m$.

For example, $m=1$ corresponds to the original Hamiltonian where $D_1=0,U_0=I$, and $V_1,E$ are the terms of $V$ with small (large) supports. At the final step $m=m_*$ with $m_*$ to be determined, \eqref{eq:UkHUk=} becomes \eqref{eq:UHU=star} 
where \begin{equation}\label{eq:final_step}
    U=U_{m_*-1}, \quad D_*=D_{m_*}-c_*I, \quad V_*=V_{m_*}+E_{m_*}.
\end{equation}
with a suitable number $c_*$.

To proceed to the next $(m+1)$-th order, we let $V=V_m$ in \eqref{eq:H0A+V=PV} and let the corresponding $A$ be $A_m$. With \begin{equation}\label{eq:U=UA}
    U_m = U_{m-1}\ee^{A_m},
\end{equation}
we have\begin{equation}\label{eq:Hk+V=Hk+1}
    \ee^{-\mathcal{A}_m}(H_0 + D_m+ V_m+E_m) = H_0 + D_{m+1}+ V_{m+1}+E_{m+1}, \where \mathcal{A}_m = [A_m, \cdot].
\end{equation}
Here \begin{equation}\label{eq:D=D+PV}
    D_{m+1} := D_m + \PP V_m,
\end{equation}
and $V_{m+1}, E_{m+1}$ are defined as follows: first define \begin{align}\label{eq:Vk+1}
    V'_{m+1} &=  \ee^{-\cA_m}(H_0 + D_m + V_m) - H_0 - (D_m+ \PP V_m) \nonumber\\ 
    &= (\ee^{-\cA_m} - 1)H_0 + \PP^\perp V_m + (\ee^{-\cA_m}-1)(D_m+V_m)  \nonumber\\
    &= -\int_0^1 \dd s \lr{ \ee^{-s\cA_m}-1}\PP^\perp V_m +(\ee^{-\cA_m}-1)(D_m+V_m), \nonumber\\
    &= [D_m, A_m]+ (\ee^{-\cA_m}-1)V_m+ \int_0^1 \dd s \lr{ \ee^{-s\cA_m}-1}\lr{[D_m,A_m]-\PP^\perp V_m},
\end{align}
with local decomposition $V'_{m+1,S}$. 
Then \begin{subequations}
    \begin{align}
        V_{m+1} &= \sum_{S:|S|<d_{\rm s}} V'_{m+1,S}, \label{eq:Vm+1=} \\
        E_{m+1} &= \ee^{-\cA_m}E_m + \sum_{S:|S|\ge d_{\rm s}} V'_{m+1,S}. \label{eq:Em+1=}
    \end{align}
\end{subequations} 

Let us first estimate the orders. Assuming $D_m\sim \epsilon$ and $V_m\sim A_m\sim v_m \ll \epsilon$, the last two terms are subdominant because they are at least second order in $v_m$. It is thus crucial to bound the first term. Naively, \begin{equation}\label{eq:finalterm_naive}
    [D_m, A_m] \sim \epsilon v_m\times \text{(support volume of $A_m$)}\sim \epsilon v_m m,
\end{equation}
which becomes bad already at a $n$-independent order $m\sim 1/\epsilon$; note that this is also the key problem for obtaining tight bounds on the lifetime of false vacuum decay \cite{our_metastable}. One has to use the fact that code $H_0$ has large distance to control this term.

Actually, the naive estimate \eqref{eq:finalterm_naive} for the $\kappa$-norm of the whole term $[D_m, A_m]$ cannot be improved. Here is the simplest example: Consider the 2d toric code where $D_m$ can possibly be $\epsilon\sum_j X_j(I-Q_{\chec_j})$ where $\chec_j$ is a check near $j$. Then consider a $Z$-string $A_m=Z_1Z_2\cdots Z_L$ (or more precisely, the corresponding $A$ for $V=Z_1Z_2\cdots Z_L$ in \eqref{eq:H0A+V=PV}), where $L$ is much smaller than the system size (code distance). Then $\norm{[D_m, A_m]}\propto L$ because every $X_j$ with $j\in\{1,\cdots,L\}$ anti-commutes with $A_m$. This is different to the false vacuum decay setting \cite{our_metastable} where $\norm{[D_m, A_m]}_\kappa$ can be bounded in a tighter way by detailed analysis.

The resolution for the above toric code example is that, the extensive $\propto L$ part of $[D_m, A_m]$ annihilates the ground states by the $I-Q_{\chec_j}$ factor, so is already block diagonal. Only the contributions from the two ends of the string may not be block diagonal (where $C_j$ anti-commutes with $A_m$; this is not possible in the middle of the string), which should be further rotated away in the next order of SWT.

\subsubsection{Locality bounds}
As suggested by the above estimates, we should keep track of two norms of the perturbation: \begin{equation}\label{eq:vm=}
    v_m:=\norm{V_m}_{\kappa_m},\quad \tv_m:=\norm{\PP^\perp V_m}_{\kappa_m}\le v_m,
\end{equation}
where $v_m$ may be parametrically larger than $\tv_m$. $\kappa_m$ will be defined shortly in \eqref{eq:kappak=}, which decays with $m$ and controls the growth of operator support when going to higher order.  
From \eqref{eq:A=PVQH} and the gap of $H_0$, we obtain \begin{equation}\label{eq:A<V}
    \norm{A_m}_{\kappa_m} \le \tv_m,
\end{equation}
bounded by the (eventually at large enough $m$) smaller (as we will prove) of the two norms. To bound the first term of \eqref{eq:Vk+1}, we invoke the following Lemma, proved in Appendix \ref{sec:DA<}: 
\begin{lem}\label{lem:DA<}
    For any $\kappa>\kappa'$ with $\delta\kappa=\kappa-\kappa'$, \begin{equation}\label{eq:DA<DA}
        \norm{[D,A]}_{\kappa'} \le \frac{2}{\delta\kappa} \norm{D}_\kappa \norm{A}_{\kappa}.
    \end{equation}
    Moreover, if $D=\sum_{S:|S|<d_{\rm s}} D_S$ where each $D_S$ is block diagonal, $A$ is given by the solution \eqref{eq:A=PVQH} of \eqref{eq:H0A+V=PV} for a given $V$, and the soundness function $f$ grows slowly such that \eqref{eq:sum_r<} holds, then \begin{equation}\label{eq:PDA<DV}
        \norm{\PP^\perp[D,A]}_{\kappa'} \le \lr{\frac{1}{\delta\kappa}+2\Delta c_f'' \ee^{c_f'(\delta\kappa)^{-\alpha}}} \norm{D}_\kappa \norm{\PP^\perp V}_{\kappa'} \le \widetilde{c}_f'' \ee^{c_f'(\delta\kappa)^{-\alpha}} \norm{D}_\kappa \norm{\PP^\perp V}_{\kappa'},
    \end{equation}
    where constant $\widetilde{c}_f''\ge 2$ is determined by $\Delta,c_f',c_f'',\alpha$.
\end{lem}
Here \eqref{eq:DA<DA} is just a combinatorial result; the crucial bound is \eqref{eq:PDA<DV} bounding the off-block-diagonal part, where $\kappa'$ for $\PP^\perp V$ indicates that the volume factor in \eqref{eq:finalterm_naive} is gone. \eqref{eq:PDA<DV} can be viewed as a generalization of Corollary 3 in \cite{topo_Hastings}, which comes from the fact that $D$ is ``relatively bounded'' by $H_0$. In comparison, our proof of Lemma \ref{lem:DA<} (presented later) is arguably simpler as we do not need the relative boundedness explicitly; we get rid of this by keeping track of the extra syndrome information in decompositions \eqref{eq:O=OSs}. Nevertheless, we will need the relative boundedness for showing that the gap does not close for $H_0+D_{m_*}$, after doing SWTs.

To bound the other terms in \eqref{eq:Vk+1}, we do not need the block diagonal property of $D$ and just invoke the following locality bound based on cluster expansions:
\begin{lem}[Lemma 23 in \cite{our_metastable}]\label{lem:clustexp}
Consider two operators $A,\OO$. If \begin{equation}\label{eq:AK<dK}
    \norm{A}_\kappa \le \frac{\delta \kappa}{3}:= \frac{\kappa-\kappa'}{3},
\end{equation}
then \begin{subequations}\label{eq:clustexp}
    \begin{align}\label{eq:clustexp1}
        \norm{\lr{\ee^{-\cA}-1}\OO}_{\kappa'}, 2\norm{\int_0^1 \dd s \lr{ \ee^{-s\cA}-1}\OO}_{\kappa'} &\le \frac{18}{\kappa' \delta \kappa} \norm{A}_\kappa \norm{\OO}_\kappa, \\
        \norm{\ee^{-\cA}\OO}_{\kappa'} &\le \lr{1+\frac{18}{\kappa' \delta \kappa} \norm{A}_\kappa} \norm{\OO}_\kappa. \label{eq:clustexp2}
    \end{align}
\end{subequations} 
\end{lem}
See Lemma 4.1 in \cite{abanin2017rigorous} for the proof of the situation where the $\kappa$ norm does not keep track of the syndrome information, and \cite{our_metastable} for the generalization to the case with syndromes. Note that although \cite{our_metastable,abanin2017rigorous} focus on constant dimensional lattices where operators $\OO=\sum_S \OO_S$ are decomposed to only connected subsets $S$ on the lattice, the above Lemma therein holds generally for any interaction graph where $S$ can be any subset. This can be checked through the proof in \cite{abanin2017rigorous}, and can be simply inferred from the fact that Lemma \ref{lem:clustexp} does not involve constants like $\Delta$ about the underlying graph.

\subsubsection{Obtaining the flow equations}

We choose a slow decay of $\kappa_m$ following \cite{abanin2017rigorous,our_metastable}: \begin{equation}\label{eq:kappak=}
    \kappa_m:=\frac{\kappa_1}{2}\lr{1+\frac{1}{1+\log m}} \ge \frac{\kappa_1}{2},
\end{equation}
which is consistent at $m=1$, and leads to \begin{subequations}\label{eq:deltakap=} 
\begin{align}
    \delta \kappa_m:=&\kappa_{m}-\kappa_{m+1}= \frac{\kappa_1\log((m+1)/m)}{2(1+\log m)[1+\log(m+1)]}\ge \frac{\kappa_1}{6m (1+\log^2 m)}\ge \frac{\kappa_1}{6m^2}, \\ 
\label{eq:deltakap'=} 
    \widetilde{\delta\kappa}_m:=&\kappa_{m}-\kappa_{2m}= \frac{\kappa_1\log 2}{2(1+\log m)[1+\log(2m))]}\ge \frac{\kappa_1}{5 (1+\log m)^2},
\end{align}
\end{subequations}
for any $m\ge 1$. The $6$ factor in \eqref{eq:deltakap=} comes from numerically plotting the function, and the factor $5$ comes from $2(1+\log 2)/\log 2<5$. We will frequently use the fact that $\delta \kappa_m\ge\delta \kappa_{m'} $ for $m\le m'$.

In the following, we bound the locality of $V'_{m+1}$ in \eqref{eq:Vk+1}, and furthermore $v_{m+1}$ and $\tv_{m+1}$,  using the above two Lemmas.

For the first term in \eqref{eq:Vk+1}, we do not directly plug $D=D_m$ in \eqref{eq:PDA<DV} due to the $\delta\kappa$ dependence: $\kappa=\kappa_m$ decays with $m$ and will ultimately make the factors in front of $\norm{\PP^\perp V}_{\kappa'}$ larger than $1$ at an order $m$ independent of system size, where SWT would have to stop. Instead, we know from the iteration \eqref{eq:D=D+PV} that \begin{equation}\label{eq:D=PV+PV}
    D_m = \PP V_1 + \cdots + \PP V_{m-1}= \lr{\PP V_1 + \cdots + \PP V_{\floor{m/2}}} + \lr{\PP V_{\floor{m/2}+1}+\cdots+\PP V_{m-1}},
\end{equation}
where each term is locally block-diagonal, and controlled by a different $\kappa$. We have also separated $D_m$ into two parts, where we want to use \eqref{eq:PDA<DV} for the $m'\le \floor{m/2}$ part, and \eqref{eq:DA<DA} for the other part. More precisely, we bound the block-off-diagonal part of the first term in \eqref{eq:Vk+1} by \begin{align}\label{eq:PDA<vd}
    &\norm{\PP^\perp[D_m,A_m]}_{\kappa_{m+1}} \le \sum_{m'=1}^{m-1} \norm{\PP^\perp[\PP V_{m'},A_m]}_{\kappa_{m+1}} \nonumber\\
    &\quad \le \sum_{m'=1}^{\floor{m/2}} \widetilde{c}_f'' \ee^{c_f'(\kappa_{m'}-\kappa_{m+1})^{-\alpha}} \norm{\PP V_{m'}}_{\kappa_{m'} }\norm{\PP^\perp V_m}_{\kappa_{m+1}} + \sum_{m'=\floor{m/2}+1}^{m-1} \frac{2}{\kappa_{m}-\kappa_{m+1}}\norm{\PP V_{m'}}_{\kappa_{m} }\norm{\PP^\perp V_m}_{\kappa_{m}} \nonumber\\
    &\quad \le \tv_m \lr{\sum_{m'=1}^{\floor{m/2}} \widetilde{c}_f'' \ee^{c_f'(\widetilde{\delta\kappa}_{m'})^{-\alpha}}  v_{m'}+\frac{2}{\delta\kappa_m}\sum_{m'=\floor{m/2}+1}^{m-1} v_{m'} }= \tv_m \lr{\mathbbm{d}_m+ \frac{1}{\delta\kappa_m}\td_m } ,
\end{align}
where we introduce two effective local norms of $D_m$: \begin{equation}\label{eq:dm=}
    \mathbbm{d}_m := \sum_{m'=1}^{\floor{m/2}} \widetilde{c}_f'' \ee^{c_f'(\widetilde{\delta\kappa}_{m'})^{-\alpha}}  v_{m'},\quad \td_m := 2\sum_{m'=\floor{m/2}+1}^{m-1} v_{m'}.
\end{equation}
On the other hand, the whole first term in \eqref{eq:Vk+1} is bounded by \eqref{eq:DA<DA}: \begin{align}\label{eq:DA<vv}
    \norm{[D_m,A_m]}_{\kappa_{m+1}} &\le \frac{2}{\delta\kappa_m} \norm{D_m}_{\kappa_m} \norm{A_m}_{\kappa_m} \le \frac{2}{\delta\kappa_m}\tv_m \sum_{m'=1}^{m-1}\norm{\PP V_{m'}}_{\kappa_m} \nonumber\\
    &\le \frac{2}{\delta\kappa_m}\tv_m \sum_{m'=1}^{m-1} v_{m'} \le \frac{1}{\delta\kappa_m}\tv_m \lr{\mathbbm{d}_m+\td_m},
\end{align}
where we have used $\mathbbm{d}_m\ge 2\sum_{m'=1}^{\floor{m/2}}v_{m'}$ from $\widetilde{c}_f''\ge 2$. 

For now, assume the condition of Lemma \ref{lem:clustexp}, \begin{equation}\label{eq:tv<dkappa}
    \tv_m \le \delta\kappa_m/3,
\end{equation}
holds here (we will verify this condition indeed holds by induction later). Then
Lemma \ref{lem:clustexp} bounds the rest of the terms in \eqref{eq:Vk+1} that we denote by $V''_{m+1}:=V'_{m+1}-[D_m,A_m]$: \begin{align}\label{eq:PV''<}
    \norm{\PP^\perp V''_{m+1}}_{\kappa_{m+1}} \le \norm{ V''_{m+1}}_{\kappa_{m+1}} &\le \frac{9}{\kappa_{m+1}\delta\kappa_m}\norm{A_m}_{\kappa_m} \lr{2\norm{V_m}_{\kappa_m}+2\norm{[D_m,A_m]}_{\frac{\kappa_m+\kappa_{m+1}}{2}} + \norm{\PP^\perp V_m}_{\kappa_m}} \nonumber\\
    &\le \frac{9}{\kappa_{m+1}\delta\kappa_m} \tv_m \lr{3v_m + \frac{4}{\delta\kappa_m}\tv_m \lr{\mathbbm{d}_m+\td_m}}. 
\end{align}
Here we have used \eqref{eq:PP<1}, \eqref{eq:vm=} and that
$\norm{[D_m,A_m]}_{\frac{\kappa_m+\kappa_{m+1}}{2}}$ is bounded by \eqref{eq:DA<vv} with an extra factor of $2$ due to the different $\kappa$-norm.

The perturbation at next order satisfies \begin{equation}\label{eq:vm+1<V'}
    v_{m+1}=\norm{V_{m+1}}_{\kappa_{m+1}} \le \norm{V'_{m+1}}_{\kappa_{m+1}} ,\quad \tv_{m+1} =\norm{\PP^\perp V_{m+1}}_{\kappa_{m+1}}\le \norm{\PP^\perp V'_{m+1}}_{\kappa_{m+1}},
\end{equation}
because $V_{m+1}$ is just the support $<d_{\rm s}$ terms of $V'_{m+1}$.
Combining \eqref{eq:Vk+1}, \eqref{eq:PDA<vd}, \eqref{eq:DA<vv}, \eqref{eq:PV''<} and \eqref{eq:vm+1<V'}, we get a set of iteration relations (so called flow equations \cite{topo_Hastings}) for variables $(v_m, \tv_m, \mathbbm{d}_m,\td_m)$ with $m\ge 1$: \begin{subequations}\label{eq:iterate}
    \begin{align}
        v_{m+1} &\le \frac{1}{\delta\kappa_m} \mathbbm{d}_m \tv_m + \frac{1}{\delta\kappa_m} \td_m \tv_m + \frac{9}{\kappa_{m+1}\delta\kappa_m} \tv_m \lr{3v_m + \frac{4}{\delta\kappa_m}\tv_m \lr{\mathbbm{d}_m+\td_m}}, \label{eq:v<dtv} \\
        \tv_{m+1} &\le \mathbbm{d}_m \tv_m + \frac{1}{\delta\kappa_m} \td_m \tv_m + \frac{9}{\kappa_{m+1}\delta \kappa_m} \tv_m \lr{3v_m + \frac{4}{\delta\kappa_m}\tv_m \lr{\mathbbm{d}_m+\td_m}}, \label{eq:tv<dtv} \\
        \mathbbm{d}_{m+1} &\le \mathbbm{d}_m + \mathbb{I}[m \text{ is odd}]\, \widetilde{c}_f'' \ee^{c_f'(\widetilde{\delta\kappa}_{(m+1)/2})^{-\alpha}} v_{(m+1)/2}, \label{eq:d<d+tv} \\
        \td_m &= 2\sum_{m'=\floor{m/2}+1}^{m-1} v_{m'}.
    \end{align}
\end{subequations}
Here $\mathbb{I}[m \text{ is odd}]$ is an indicator function that returns $0,1$ for $m$ being even/odd. 
The initial values are \begin{equation}\label{eq:iterate_initial}
    v_1 , \tv_1\le \epsilon ,\quad \mathbbm{d}_1=\td_1=\td_2 = 0.
\end{equation}

\subsubsection{The last SWT step}

We analyze the iteration problem above by the following Lemma that we prove in Appendix \ref{sec:iterate}: 
\begin{lem}\label{lem:iterate}
    There exist constants $c_{\rm iter},\epsilon_0$ determined by $\Delta,\kappa_1$ such that \begin{equation}
        c_{\rm iter} \epsilon_0 \le 1/4. \label{eq:citereps<} 
    \end{equation}
    Furthermore, for any $\epsilon\le \epsilon_0$, the iteration \eqref{eq:iterate} with $\kappa_m$ defined in \eqref{eq:kappak=} and initial value \eqref{eq:iterate_initial} generates variables that satisfy \begin{equation}\label{eq:vm<epsm}
    v_m\le \frac{\epsilon}{\delta\kappa_{m-1}} (c_{\rm iter}\epsilon)^{m-1},\quad \tv_m \le \epsilon\, (c_{\rm iter}\epsilon)^{m-1}, \quad \mathbbm{d}_m \le \frac{2}{3}c_{\rm iter}\epsilon,\quad \td_m \le \frac{3\epsilon}{\delta\kappa_{m-2}} (c_{\rm iter}\epsilon)^{\floor{m/2}}
\end{equation}
for any $m\ge 2$. The condition \eqref{eq:tv<dkappa} required to derive the flow equations is also satisfied for all $m\ge 1$.
\end{lem}

This Lemma indicates that the iteration can go to $m\rightarrow \infty$. However, we truncate at some order $m_*$ since there is no need to proceed once the perturbation $V_{m_*}$ is smaller than the garbage term $E_{m_*}$, which we cannot rotate away via SWT. We do not care about the local decomposition \eqref{eq:O=OSs} of $E_m$, and just bound its total operator norm. To this end, we first obtain \begin{align}\label{eq:V>d<}
    \norm{\sum_{S:|S|\ge d_{\rm s}} V'_{m+1,S}} &\le \sum_{i\in \Lambda}\sum_{S\ni i:|S|\ge d_{\rm s}}\sum_\bs\norm{V'_{m+1,S,\bs}} \le \sum_{i\in \Lambda}\ee^{-\kappa_{m+1}d_{\rm s}}\sum_{S\ni i:|S|\ge d_{\rm s}}\sum_\bs\norm{V'_{m+1,S,\bs}}\ee^{\kappa_{m+1}|S|} \nonumber\\
    &\le n\,\ee^{-\kappa_{m+1}d_{\rm s}}\norm{V'_{m+1}}_{\kappa_{m+1}}\le n\,\ee^{-\kappa_{m+1}d_{\rm s}}\cdot \frac{\epsilon}{\delta\kappa_m} (c_{\rm iter}\epsilon)^{m}\le n\frac{6\epsilon}{ \kappa_1}m^2 (c_{\rm iter}\epsilon)^{m}\ee^{-\kappa_1 d_{\rm s}/2},
\end{align}
for $m\ge 1$,
where we have used \eqref{eq:kappak=}, \eqref{eq:deltakap=}, and the bound \eqref{eq:vm<epsm} on $v_{m+1}$ that also applies to $\norm{V'_{m+1}}_{\kappa_{m+1}}$. Similarly, $\norm{E_1}\le n\,\ee^{-\kappa_{1}d_{\rm s}}\epsilon$. 
Then \eqref{eq:Em+1=} leads to 
\begin{align}\label{eq:Em*<summ}
    \norm{E_{m_*}} &\le \norm{E_1} + \sum_{m=1}^{m_*-1}\norm{\sum_{S:|S|\ge d_{\rm s}} V'_{m+1,S}}\le \max\lr{1,\frac{6}{\kappa_1}} n\epsilon \ee^{-\kappa_1d_{\rm s}/2} \sum_{m=0}^{m_*-1}\max(1,m^2) (c_{\rm iter}\epsilon)^m  \nonumber\\
    &\le \frac{c_2}{2} n \epsilon\, \ee^{-\kappa_1d_{\rm s}/2}, \where c_2:=4\max\lr{1,\frac{6}{\kappa_1}},
\end{align}
independent of $m_*$.
Here we have used \eqref{eq:citereps<} and $\sum_{m=0}^\infty \max(1,m^2) 4^{-m} < 2$.
Since $\norm{V_{m_*}}\le n v_{m_*}$ with $v_m$ bounded in \eqref{eq:vm<epsm}, we stop the SWTs at step $m_*=\Theta(d_{\rm s})$ that is the smallest integer fulfilling the following: \begin{equation}\label{eq:mstar=}
    \frac{6m_*^2}{\kappa_1} (c_{\rm iter}\epsilon)^{m_*-1} \le \ee^{-\kappa_1d_{\rm s}/2}, \rarrow \norm{V_{m_*}}\le  n\frac{\epsilon}{\delta\kappa_{m_*-1}} (c_{\rm iter}\epsilon)^{m_*-1}\le \frac{c_2}{2} n \epsilon\, \ee^{-\kappa_1d_{\rm s}/2}, 
\end{equation} 
The total garbage in the end is thus \begin{equation}
    \norm{V_{m_*}+E_{m_*}} \le c_2 n \epsilon\, \ee^{-\kappa_1d_{\rm s}/2}.
\end{equation}
This establishes the first property \eqref{eq:V*<delta} we want with $c_2$ determined by $\kappa_1$. 

For the second property of relative boundedness, we follow \cite{topo_Hastings} to show the following, with the proof delayed until Appendix \ref{sec:rela_bound}: 
\begin{prop}\label{prop:rela_bound}
    Suppose the soundness function $f$ grows slowly such that \eqref{eq:sum_r<} holds. If $D=\sum_{S:|S|<d_{\rm s}}\sum_\bs D_{S,\bs}$ where each $D_{S,\bs}$ is block diagonal, then there exists a number $c_D$ such that $D-c_DI$ is relatively bounded by $H_0$: for any state $\ket{\psi}$, \begin{equation}\label{eq:rela_bound}
        \norm{(D-c_DI)\ket{\psi}} \le \lr{2+\frac{\Delta}{ \kappa}} \Delta c_f'' \ee^{c_f'(\kappa/2)^{-\alpha}} \norm{D}_\kappa  \norm{H_0\ket{\psi}}.
    \end{equation}
\end{prop}

Letting $D=D_{m_*}$ in Proposition \ref{prop:rela_bound}, there exist constant $c_*$ such that for any state $\ket{\psi}$,
\begin{align}
    \norm{D_*\ket{\psi}} &\le \lr{2+\frac{\Delta}{ \kappa_1/2}} \Delta c_f'' \ee^{c_f'(\kappa_1/4)^{-\alpha}} \norm{D}_{\kappa_1/2}  \norm{H_0\ket{\psi}} \nonumber\\
    &\le 2 \lr{1+\frac{\Delta}{ \kappa_1}} \Delta c_f'' \ee^{c_f'(\kappa_1/4)^{-\alpha}}   \norm{H_0\ket{\psi}} \sum_{m=1}^{m_*-1} \norm{\PP V_m}_{\kappa_1/2} \nonumber\\
    &\le 2 \lr{1+\frac{\Delta}{ \kappa_1}} \Delta c_f'' \ee^{c_f'(\kappa_1/4)^{-\alpha}}   \norm{H_0\ket{\psi}} \sum_{m=1}^{m_*-1} v_m \nonumber\\
    &\le 2 \lr{1+\frac{\Delta}{ \kappa_1}} \Delta  \ee^{c_f'(\kappa_1/4)^{-\alpha}}   \norm{H_0\ket{\psi}} \lr{\mathbbm{d}_{m_*} + c_f''\td_{m_*}/2} \nonumber\\
    &\le c_1 \epsilon \norm{H_0\ket{\psi}},
\end{align}
where $D_*$ is defined in \eqref{eq:final_step}.
Here we have used \eqref{eq:D=PV+PV} in the second line, \eqref{eq:PP<1} and $\kappa_m\ge \kappa_1/2$ in the third line, and \eqref{eq:dm=} in the fourth line. The last line comes from \eqref{eq:vm<epsm} and \eqref{eq:citereps<} with constant \begin{equation}
    c_1 = 2 \lr{1+\frac{\Delta}{ \kappa_1}} \Delta  \ee^{c_f'(\kappa_1/4)^{-\alpha}} \lr{\frac{2}{3}c_{\rm iter} + c_f'' \max_{m\ge 3}\frac{3}{\delta\kappa_{m-2}} 4^{-\floor{m/2}}},
\end{equation}
where the maximum over $m$ is bounded because of the exponential decay in $m$.
The second property \eqref{eq:D*<H0*} is thus also established.

Finally, it remains to prove \eqref{eq:At<}. Since $U=U_{m_*-1}=\ee^{A_1}\cdots \ee^{A_{m_*-1}}$ from \eqref{eq:U=UA}, the time-dependent generator in \eqref{eq:U=intA} can be chosen as
\begin{equation}\label{eq:As=Ak}
    A(t) = \left\{ \begin{aligned}
        &0, & &t< 2^{1-m_*} \\
        &2^m A_m, & 2^{-m} \le \, & t < 2^{1-m}, \quad (m=1,\cdots, m_*-1)
    \end{aligned} \right.
\end{equation}
As a result, for any subset $S$, \eqref{eq:At<} is satisfied: \begin{align}
    \max_t \norm{A(t)}_{\kappa_1/2} &\le \max_{m=1,\cdots,m_*-1} 2^m \norm{A_m}_{\kappa_1/2}\le \max_{m=1,\cdots,m_*-1} 2^m \norm{A_{m}}_{\kappa_m} \nonumber\\
    &\le \max_{m=1,\cdots,m_*-1} 2^m \tv_m \le \max_{m=1,\cdots,m_*-1} 2^m \epsilon 4^{1-m} = 2\epsilon.
\end{align}
Here in the second line, we have used \eqref{eq:A<V} and \eqref{eq:vm<epsm} with \eqref{eq:citereps<}. 

We have established the three conditions for Lemma \ref{lem:spec}, so this finishes the proof of Theorem \ref{thm:qtm}.

\subsection{Putting everything together: Proof of Lemma \ref{lem:spec}} \label{sec:spec}
\begin{proof}[Proof of Lemma \ref{lem:spec}]
Due to the gap of $H_0$ and relative boundedness \eqref{eq:D*<H0*}, Lemma 2 in \cite{topo_Hastings} implies that the spectrum of $H_0+D_*$ is contained in \begin{equation}\label{eq:H0+D*_spec}
    \{0\} \cup [1-c_1\epsilon ,\infty)\quad \subset\quad \{0\} \cup [2/3,\infty),
\end{equation}
where the exact energy-$0$ subspace consists of the unperturbed $2^k$ ground states in $P$, and we have used $1-c_1\epsilon\ge 2/3$ from \eqref{eq:V=eps}.

We then consider $V_*$ as a perturbation to $H_0+D_*$. According to Weyl's inequality, \begin{equation}
    \abs{\lambda_m(A+B)-\lambda_m(A)} \le \norm{B},
\end{equation}
where $\lambda_m(A)$ is the $m$-th smallest eigenvalue of $A$. Letting $A=H_0+D_*$ and $B=V_*$, the spectrum of $U^\dagger H U$ is thus contained in the two disjoint intervals \eqref{eq:spec_H} due to \eqref{eq:H0+D*_spec}, with exactly $2^k$ ground states. Here to bound the final gap by $1/2$, we have used \begin{equation}
    3\epsilon_*\le 1/6,
\end{equation}
which is achievable if $n$ is larger than a constant $c_3$ because $\epsilon_*$ decays with $n$ in \eqref{eq:eps*=}. The spectrum of $H$ is thus also contained in \eqref{eq:spec_H}, since conjugation by a unitary does not change eigenvalues.

To prove stability of eigenvectors \eqref{eq:Pnew-P}, we first show that for any normalized new ground state \begin{equation}\label{eq:psi=Pnew}
    \ket{\psi}=P_{\rm new}'\ket{\psi}, \where P_{\rm new}':=\lr{U^\dagger P_{\rm new}U}
\end{equation}
(here we work in the rotated frame with the new ground states contained in $U^\dagger P_{\rm new}U$), its support on the unperturbed excited-state subspace $I-P$ of $H_0+D_*$ is small: \begin{equation}\label{eq:I-P<delta}
    \norm{(I-P)\ket{\psi} } \le 2\sqrt{\epsilon_*}.
\end{equation}
This comes from a Markov inequality for evaluating the total energy \begin{equation}
    \epsilon_* \ge \bra{\psi}(H_0+D_*+V_*)\ket{\psi} \ge \bra{\psi}(H_0+D_*)\ket{\psi} - \norm{V_*} \ge (1-c_1\epsilon)\bra{\psi}(I-P)\ket{\psi} - \epsilon_* \ge \frac{1}{2}\norm{(I-P)\ket{\psi} }^2-\epsilon_*,
\end{equation}
where the first and the third inequality comes from the spectrums \eqref{eq:spec_H} and \eqref{eq:H0+D*_spec}, respectively. We have also used $1-c_1\epsilon\ge 1/2$ from \eqref{eq:V=eps}. \eqref{eq:I-P<delta} actually implies \begin{equation}\label{eq:P-PP<}
    \norm{(I-P)P_{\rm new}'}=\norm{(P_{\rm new}'-P)P_{\rm new}'}\le 2\sqrt{\epsilon_*},
\end{equation}
due to \eqref{eq:psi=Pnew}.

On the other hand, we can view $H_0+D_*$ as perturbing $H_0+D_*+V_*$ by $-V_*$. For any original ground state $\ket{\psi}=P\ket{\psi}$, its support in the new excited states is also small: \begin{equation}\label{eq:I-P<delta1}
    \norm{(I-P_{\rm new}')\ket{\psi} } \le 2\sqrt{\epsilon_*},
\end{equation}
because \begin{align}
    0 = \bra{\psi}(H_0+D_*)\ket{\psi} \ge \bra{\psi}(H_0+D_*+V_*)\ket{\psi} - \norm{V_*} &\ge (1-c_1\epsilon-\epsilon_*)\bra{\psi}(I-P_{\rm new}')\ket{\psi}- \epsilon_* \bra{\psi}P_{\rm new}'\ket{\psi} - \epsilon_* \nonumber\\
    &\ge \frac{1}{2} \norm{(I-P_{\rm new}')\ket{\psi}}^2 - 2\epsilon_*.
\end{align}
\eqref{eq:I-P<delta1} leads to $\norm{(P_{\rm new}'-P)P}\le 2\sqrt{\epsilon_*}$, so that combining with \eqref{eq:P-PP<}, we have \begin{align}
    \norm{P_{\rm new}'-P}&=\norm{(P_{\rm new}'-P)(P+I-P)} \le \norm{(P_{\rm new}'-P)P} + \norm{(P_{\rm new}'-P)(I-P)} \nonumber\\
    &= \norm{(P_{\rm new}'-P)P} + \norm{P_{\rm new}'(I-P)} \le 4\sqrt{\epsilon_*}.
\end{align}
Here to the second line comes from $P(I-P)=0$ and plugging in the bounds above.

Therefore, Theorem \ref{thm:qtm} holds since \eqref{eq:At<} is also satisfied by assumption.
\end{proof}

\subsection{Analyzing the flow equations: Proof of Lemma \ref{lem:iterate}}\label{sec:iterate}
\begin{proof}[Proof of Lemma \ref{lem:iterate}]
First observe that if the first bound in \eqref{eq:vm<epsm} on $v_m$ holds, then the last bound on $\td_m$ follows because \begin{align}
    \td_m &= 2\sum_{m'=\floor{m/2}+1}^{m-1} v_{m'} \le 2\sum_{m'=\floor{m/2}+1}^{m-1} \frac{\epsilon}{\delta\kappa_{m'-1}} (c_{\rm iter}\epsilon)^{m'-1} \le \frac{2\epsilon}{\delta\kappa_{m-2}}\sum_{m'=\floor{m/2}+1}^{m-1}(c_{\rm iter}\epsilon)^{m'-1} \nonumber\\
    &\le \frac{2\epsilon}{\delta\kappa_{m-2}} (c_{\rm iter}\epsilon)^{\floor{m/2}} \lr{1+\frac{1}{4}+\frac{1}{4^2}+\cdots} \le \frac{3\epsilon}{\delta\kappa_{m-2}} (c_{\rm iter}\epsilon)^{\floor{m/2}}.
\end{align}
Therefore, we can focus on the first three bounds in \eqref{eq:vm<epsm}.

Let \begin{equation}\label{eq:citer=}
    c_{\rm iter}:=\max\lr{\frac{3\widetilde{c}_f''}{2}\lr{\ee^{c_f'(\widetilde{\delta\kappa}_1)^{-\alpha}}+ \sum_{m=2}^\infty \ee^{c_f'(\widetilde{\delta\kappa}_{m})^{-\alpha}}  \frac{4^{1-m}}{\delta\kappa_{m-1}}  },\frac{27}{\kappa_2 \delta\kappa_1}\max\lr{1, \delta\kappa_1} },
\end{equation}
be a constant determined by $c_f',\widetilde{c}_f'',\alpha,\kappa_1$. In particular, the infinite sum converges because the exponential decay $4^{1-m}$ dominates e.g. growth of $\ee^{c_f'(\widetilde{\delta\kappa}_{m})^{-\alpha}}= \ee^{\mathrm{O}[(\log m)^{2\alpha}]}$ from \eqref{eq:deltakap'=}.
There exists a constant \begin{equation}\label{eq:eps0<deltakap}
    0<\epsilon_0<\min\lr{1/(4c_{\rm iter}), \delta\kappa_1/3},
\end{equation} 
determined by $\Delta,\kappa_1$ such that $\forall m\ge 2$, \begin{equation}
    \frac{3}{\delta\kappa_m\delta\kappa_{m-1}} (c_{\rm iter}\epsilon_0)^{\floor{m/2}}+ \frac{9}{\kappa_{m+1}\delta \kappa_m} (c_{\rm iter}\epsilon_0)^{m-1} \lr{\frac{3}{\delta\kappa_{m-1}}+ \frac{4\epsilon_0}{\delta\kappa_m}\lr{\frac{2}{3}c_{\rm iter}+ \frac{3}{\delta\kappa_{m-1}}}} \le \frac{c_{\rm iter}}{3} \min\lr{1, \frac{1}{ \delta\kappa_2}}. \label{eq:citerm<}
\end{equation}
The reason is that \eqref{eq:citerm<} is always suppressed by a positive power of $\epsilon_0$, and
this exponential-decaying factor $(c_{\rm iter}\epsilon_0)^{m-1}$ dominates over the polynomial dependence of $\kappa_{m+1},\delta \kappa_{m'}$ ($m'=m,m-1$) at large $m$, so there always exists a sufficiently small $\epsilon_0$ that does the job. 


For $m=2$, \eqref{eq:tv<dtv} yields \begin{equation}
    \tv_2 \le 0 + \frac{9}{\kappa_2 \delta\kappa_1} \epsilon(3\epsilon + 0) \le c_{\rm iter} \epsilon^2,
\end{equation}
according to \eqref{eq:citer=}. \eqref{eq:v<dtv} yields $v_2\le c_{\rm iter} \epsilon^2/(\delta\kappa_1)$ similarly, and \eqref{eq:d<d+tv} leads to \begin{equation}\label{eq:d2<}
    \mathbbm{d}_2\le \widetilde{c}_f'' \ee^{c_f'(\widetilde{\delta\kappa}_1)^{-\alpha}} \epsilon\le \frac{2}{3}c_{\rm iter}\epsilon,
\end{equation}
using \eqref{eq:citer=}. Therefore, \eqref{eq:vm<epsm} is satisfied at $m=2$.

We then prove \eqref{eq:vm<epsm} inductively for any $\epsilon\le \epsilon_0$ and $c_{\rm iter},\epsilon_0$ defined above: if \eqref{eq:vm<epsm} holds for all $m=1,\cdots,m'$ with $m'\ge 2$, it suffices to prove that it also holds for $m=m'+1$. For \eqref{eq:tv<dtv}, we have \begin{align}
    \tv_{m'+1} &\le \tv_{m'} \mlr{\frac{2}{3}c_{\rm iter}\epsilon+\frac{3\epsilon}{\delta\kappa_{m'}\delta\kappa_{m'-1}} (c_{\rm iter}\epsilon_0)^{\floor{m'/2}} +\frac{9\epsilon}{\kappa_{m'+1}\delta \kappa_{m'}} (c_{\rm iter}\epsilon)^{m'-1} \lr{\frac{3}{\delta\kappa_{m'-1}}+ \frac{4}{\delta\kappa_{m'}} \lr{\frac{2}{3}c_{\rm iter}\epsilon_0+ \frac{3\epsilon_0}{\delta\kappa_{m'-1}}}} } \nonumber\\
    &\le \tv_{m'} \lr{\frac{2}{3}c_{\rm iter}\epsilon+\frac{1}{3}c_{\rm iter}\epsilon } \le c_{\rm iter}\epsilon \tv_{m'} \le \epsilon \lr{c_{\rm iter}\epsilon}^{m'},
\end{align}
where we have used \eqref{eq:vm<epsm} for $m=m'$, $\delta\kappa_{m'-2}\le \delta\kappa_{m'-1}$, and that \eqref{eq:citerm<} holds for $\epsilon_0$ replaced by any $\epsilon\le \epsilon_0$. One can similarly obtain $v_{m'+1}\le c_{\rm iter}\epsilon \tv_{m'}/(\delta\kappa_{m'})$ using \eqref{eq:v<dtv} and $\delta\kappa_{m'}\le \delta\kappa_2$, so that $v_{m'+1}$ also satisfies \eqref{eq:vm<epsm}. Finally, \eqref{eq:d<d+tv} (or the equivalent \eqref{eq:dm=}) leads to \begin{align}
    \mathbbm{d}_{m'+1} &\le \mathbbm{d}_2+ \sum_{m=2}^{\floor{(m'+1)/2}} \widetilde{c}_f'' \ee^{c_f'(\widetilde{\delta\kappa}_{m})^{-\alpha}}  \frac{1}{\delta\kappa_{m-1}} v_m \le \widetilde{c}_f'' \ee^{c_f'(\widetilde{\delta\kappa}_1)^{-\alpha}} \epsilon+\widetilde{c}_f''\epsilon \sum_{m=2}^\infty \ee^{c_f'(\widetilde{\delta\kappa}_{m})^{-\alpha}}  \frac{1}{\delta\kappa_{m-1}} 4^{1-m} \le \frac{2}{3}c_{\rm iter} \epsilon,
\end{align}
where we have used the bound on $v_m$ in \eqref{eq:vm<epsm},  $c_{\rm iter}\epsilon\le 1/4$ from \eqref{eq:eps0<deltakap}, and finally \eqref{eq:d2<} and \eqref{eq:citer=}. Therefore, \eqref{eq:vm<epsm} holds inductively.

Finally, since the condition for iteration \eqref{eq:tv<dkappa} is satisfied at $m=1$ due to \eqref{eq:eps0<deltakap}, and $\delta\kappa_m\ge \delta\kappa_1 4^{1-m}$, \eqref{eq:tv<dkappa} continues to hold for all $m\ge 1$.
\end{proof}

\subsection{Application of local indistinguishability: Proof of Lemma \ref{lem:DA<}}\label{sec:DA<}
\begin{proof}[Proof of Lemma \ref{lem:DA<}]
We assign the local terms as \begin{equation}\label{eq:DASs}
    \lr{[D, A]}_{S'',\bs''} = \sum_{\substack{S, S', \bs , \bs': \\
    S\cup S'=S'',\bs+\bs'=\bs''}} [D_{S',\bs'}, A_{S,\bs}],
\end{equation}
where each term is strongly supported in $S'\cup S$, and has syndrome $\bs+\bs'$. 

 We first prove the easier \eqref{eq:DA<DA} that does not depend on e.g. the block-diagonal structure of $D$. For a given $D_{S',\bs'}$,
\begin{align}\label{eq:sumSDA<1}
    \sum_{S,\bs} \norm{[D_{S',\bs'}, A_{S,\bs}]} \ee^{\kappa' |S'\cup S|} &\le \sum_{i\in S'} \sum_{S\ni i} \sum_{\bs} 2\norm{D_{S',\bs'}}\norm{A_{S,\bs}} \ee^{\kappa'(|S|+|S'|)} \nonumber\\
    &\le 2|S'|\norm{D_{S',\bs'}} \ee^{\kappa'|S'|} \norm{A}_{\kappa'} \nonumber\\
    &\le \frac{1}{\delta\kappa} \norm{D_{S',\bs'}} \ee^{\kappa|S'|} \norm{A}_{\kappa'}.
\end{align}
Here in the first line, we have organized the sum by the fact that for any $S$ that overlaps with $S'$ such that the commutator does not vanish, there exists $i\in S'$ such that $S\ni i$. The overlapping condition also leads to $|S'\cup S|\le |S|+|S'|-1$ ($-1$ is not important and ignored). In the second line, we have invoked the definition of $\kappa$-norm for $A$. The last line comes from $2x\le \ee^x$ for $x=\delta\kappa|S'|$. Exchanging the roles of $D$ and $A$ in the previous manipulation \eqref{eq:sumSDA<1}, we also have for a given $A_{S,\bs}$: \begin{equation}\label{eq:sumS'DA<1}
    \sum_{S',\bs'} \norm{[D_{S',\bs'}, A_{S,\bs}]} \ee^{\kappa' |S'\cup S|} \le \frac{1}{\delta\kappa} \norm{D}_{\kappa'} \norm{A_{S,\bs}}\ee^{\kappa|S|}.
\end{equation}
Observe that if a local term $[D_{S',\bs'}, A_{S,\bs}]$ acts on $i$, either $i\in S$ or $i\in S'$. As a result, \eqref{eq:DA<DA} follows by \begin{align}\label{eq:commDA<}
    \norm{[D,A]}_{\kappa'} &\le \max_i \left[\sum_{S\ni i}\sum_{\bs} \sum_{S',\bs'} \norm{[D_{S',\bs'}, A_{S,\bs}]} \ee^{\kappa'|S'\cup S|} + \sum_{S'\ni i}\sum_{\bs'} \sum_{S,\bs} \norm{[D_{S',\bs'}, A_{S,\bs}]} \ee^{\kappa'|S'\cup S|}\right] \nonumber\\
    &\le \max_i \left[ \sum_{S\ni i}\sum_{\bs} \frac{1}{\delta\kappa} \norm{D}_{\kappa'} \norm{A_{S,\bs}}\ee^{\kappa|S|} + \sum_{S'\ni i}\sum_{\bs'}\frac{1}{\delta\kappa} \norm{D_{S',\bs'}} \ee^{\kappa|S'|} \norm{A}_{\kappa'}\right] \nonumber\\
    &\le \frac{1}{\delta\kappa} \lr{\norm{D}_{\kappa'} \norm{A}_{\kappa}+ \norm{D}_\kappa \norm{A}_{\kappa'}} \le \frac{2}{\delta\kappa}\norm{D}_\kappa \norm{A}_{\kappa}.
\end{align}
Here we have plugged in \eqref{eq:sumS'DA<1} and \eqref{eq:sumSDA<1} in the second line, and used e.g. $\norm{D}_{\kappa'}\le \norm{D}_{\kappa}$ in the end.

We then turn to \eqref{eq:PDA<DV}, where we need two crucial observations. 

The first is that many terms in \eqref{eq:DASs} vanish when projected by $\PP^\perp$ (see the toric code example discussed earlier for intuition): \begin{equation}\label{eq:PDA=0}
    \PP^\perp [D_{S',\bs'}, A_{S,\bs}] =0,
\end{equation}
as long as $S'$ is far from the syndrome $\bs$, compared to its own size $|S'|$. More precisely, \eqref{eq:PDA=0} holds if \begin{equation}\label{eq:S'<fdistance}
    |S'|<\widetilde{f}( \mathsf{d}(S', C)), \quad \forall \mathcal{Q}_C\in \bs,
\end{equation}
where $\widetilde{f}$ is the function determined by $f$ in \eqref{eq:tf=f}.

To see this, Lemma \ref{lem:S<tfr} implies that (here is why we restrict $D$ to only contain $|S|<d_{\rm s}$ terms) \begin{equation}\label{eq:PDP=cP}
    P_{\bar{S}'} D_{S',\bs'} P_{\bar{S}'} = c P_{\bar{S}'},
\end{equation}
for constant $c$ (note that $c=0$ for $\bs'\neq \emptyset$), where $\bar{S}'$ does not overlap with the syndrome $\bs$ according to \eqref{eq:S'<fdistance}. As a result, for $\bar{S}''=\bar{S}'\cup S$, \begin{align}\label{eq:QDAP=0}
    Q_{\bar{S}''}[D_{S',\bs'}, A_{S,\bs}] P_{\bar{S}''} &=Q_{\bar{S}''}[D_{S',\bs'}-cI, A_{S,\bs}] P_{\bar{S}''} = Q_{\bar{S}''}(D_{S',\bs'}-cI) A_{S,\bs} P_{\bar{S}''} \nonumber\\
    &= Q_{\bar{S}''}(D_{S',\bs'}-cI) A_{S,\bs}P_{\bar{S}'} P_{\bar{S}''} = Q_{\bar{S}''}(D_{S',\bs'}-cI) P_{\bar{S}'} A_{S,\bs}P_{\bar{S}''} = 0.
\end{align}
Here in the first line, we have inserted an identity $cI$ in the commutator where $c$ is the constant in \eqref{eq:PDP=cP}  such that (recall \eqref{eq:DcommutP}) \begin{equation}\label{eq:P(D-c)=0}
    P_{\bar{S}'}(D_{S',\bs'}-cI)=0,\rarrow P_{\bar{S}''}(D_{S',\bs'}-cI) =0.
\end{equation}
In the second line of \eqref{eq:QDAP=0}, we have commuted $P_{\bar{S}'}$ through $A_{S,\bs}$ because the syndrome $\bs$ does not touch $\bar{S}'$, and used \eqref{eq:P(D-c)=0} again. \eqref{eq:QDAP=0} implies that the commutator $[D_{S',\bs'}, A_{S,\bs}]$ commutes with $P_{\bar{S}''}$. On the other hand, since $[D_{S',\bs'}, A_{S,\bs}]$ is strongly supported in $S''=S'\cup S$, it commutes with checks in $P_{\bar{S}''}$ that are not contained in $S''$. As a result, in order for the commutator to commute with $P_{\bar{S}''}$, it must commute with $P_{S''}$:
\begin{equation}\label{eq:PDA=0_}
    Q_{S''}[D_{S',\bs'}, A_{S,\bs}] P_{S''} = 0, \rarrow \PP^\perp [D_{S',\bs'}, A_{S,\bs}]=Q_{S''}[D_{S',\bs'}, A_{S,\bs}] P_{S''}+\mathrm{H.c.}=0.
\end{equation}

The second observation is that $A_{S,\bs}$ is itself suppressed by a factor \begin{equation}\label{eq:A<V/s}
    \norm{A_{S,\bs}} \le \max\lr{\norm{P_S V_{S,\bs} Q_S H_S^{-1}}, \norm{H_S^{-1}Q_S V_{S,\bs} P_S }}\le \norm{P_S V_{S,\bs} Q_S}/|\bs| = \norm{\PP^\perp V_{S,\bs}}/|\bs|.
\end{equation}
Here we have used $H_S^{-1}\le 1/|\bs|$ because any state in $V_{S,\bs}P_S$ flips $|\bs|$ checks, and any flipped check cost at least energy $1$, from (\ref{eq:minlambdacheck}).

For a given $A_{S,\bs}$, these two facts lead to \begin{align}\label{eq:sumS'DA<}
    \sum_{S',\bs'} \norm{\PP^\perp [D_{S',\bs'}, A_{S,\bs}]} \ee^{\kappa' |S'\cup S|} &\le \sum_{\mathcal{Q}_C\in \bs} \sum_{i\in C} \sum_{S': |S'|\ge \widetilde{f}(\mathsf{d}(i,S'))}\sum_{\bs'} \norm{[D_{S',\bs'},A_{S,\bs}]} \ee^{\kappa'(|S|+|S'|)} \nonumber\\
    &\le \sum_{\mathcal{Q}_C\in \bs} \sum_{i\in C} \sum_{r\ge 0: \widetilde{f}(r)<d_{\rm s}} \sum_{j: \mathsf{d}(i,j)=r}\sum_{S'\ni j:|S'|\ge \widetilde{f}(r) }\sum_{\bs'} 2\norm{D_{S',\bs'}}\norm{A_{S,\bs}} \ee^{\kappa'(|S|+|S'|)} \nonumber\\
    &\le 2\norm{A_{S,\bs}}\ee^{\kappa'|S|}\sum_{\mathcal{Q}_C\in \bs} \sum_{i\in C} \sum_{r\ge 0: \widetilde{f}(r)<d_{\rm s}}\gamma(r) \norm{D}_\kappa \ee^{-\delta\kappa \widetilde{f}(r)} \nonumber\\
    &\le 2\norm{A_{S,\bs}}\ee^{\kappa'|S|}|\bs| \Delta \norm{D}_\kappa c_f'' \ee^{c_f'(\delta\kappa)^{-\alpha}} \nonumber\\
    &\le 2\Delta c_f'' \ee^{c_f'(\delta\kappa)^{-\alpha}} \norm{D}_{\kappa}\norm{\PP^\perp V_{S,\bs}} \ee^{\kappa'|S|}.
\end{align}
Here the first line comes from \eqref{eq:PP<1} and the first observation above, where for the projected commutator to not vanish, there exists $i$ contained in syndrome $\bs$ such that $|S'|\ge \widetilde{f}(\mathsf{d}(i,S'))$. This condition is further relaxed to the second line that $S'$ acts on $j\in \Lambda$ such that $|S'|\ge \widetilde{f}(\mathsf{d}(i,j))$. The sum over $r$ is restricted by $\widetilde{f}(r)<d_{\rm s}$, because $D$ only contains terms with support $|S'|<d_{\rm s}$. To get the third line of \eqref{eq:sumS'DA<}, notice that $\ee^{\kappa'|S'|}\le \ee^{\kappa|S'|}\ee^{\delta\kappa \widetilde{f}(r)}$ for any $S'$ in the sum, so we can use the $\kappa$-norm of $D$ with an exponential factor leftover. We also replaced the sum over $j$ by $\gamma(r)$, which was defined in \eqref{eq:Gamma<Delta}. We then used condition \eqref{eq:sum_r<} to get the fourth line, and \eqref{eq:A<V/s} for the last line.

For a given $D_{S',\bs'}$, we use the previous \eqref{eq:sumSDA<1} to obtain \begin{align}\label{eq:sumSDA<}
    \sum_{S,\bs} \norm{\PP^\perp [D_{S',\bs'}, A_{S,\bs}]} \ee^{\kappa' |S'\cup S|}\le \sum_{S,\bs} \norm{ [D_{S',\bs'}, A_{S,\bs}]} \ee^{\kappa' |S'\cup S|}\le \frac{1}{\delta\kappa} \norm{D_{S',\bs'}} \ee^{\kappa|S'|} \norm{\PP^\perp V}_{\kappa'},
\end{align}
where we have used $\norm{A}_{\kappa'}\le \norm{\PP^\perp V}_{\kappa'}$ (see also \eqref{eq:A<V}) because each $\norm{A_{S,\bs}}\le \norm{\PP^\perp V_{S,\bs}}$ from \eqref{eq:A<V/s} and $|\bs|\ge 1$, since $A_{S,\bs}=0$ as long as $\bs=\emptyset$.
We then follow \eqref{eq:commDA<} to derive \eqref{eq:PDA<DV}: \begin{align}
    \norm{\PP^\perp [D,A]}_{\kappa'} &\le \max_i \left[\sum_{S\ni i}\sum_{\bs} \sum_{S',\bs'} \norm{\PP^\perp[D_{S',\bs'}, A_{S,\bs}]} \ee^{\kappa'|S'\cup S|} + \sum_{S'\ni i}\sum_{\bs'} \sum_{S,\bs} \norm{\PP^\perp[D_{S',\bs'}, A_{S,\bs}]} \ee^{\kappa'|S'\cup S|} \right] \nonumber\\
    &\le \max_i \left[ \sum_{S\ni i}\sum_{\bs} 2 \Delta c_f'' \ee^{c_f'(\delta\kappa)^{-\alpha}} \norm{D}_{\kappa}\norm{\PP^\perp V_{S,\bs}} \ee^{\kappa'|S|} + \sum_{S'\ni i}\sum_{\bs'}\frac{1}{\delta\kappa} \norm{D_{S',\bs'}} \ee^{\kappa|S'|} \norm{\PP^\perp V}_{\kappa'}\right] \nonumber\\
    &\le \lr{\frac{1}{\delta\kappa}+2\Delta c_f'' \ee^{c_f'(\delta\kappa)^{-\alpha}}}\norm{D}_\kappa \norm{\PP^\perp V}_{\kappa'},
\end{align}
where the second line combines \eqref{eq:sumS'DA<} and \eqref{eq:sumSDA<}. The second expression in \eqref{eq:PDA<DV} then follows by relaxing the term $1/\delta\kappa$ to an exponential using $\alpha>0$. We have also assumed $\widetilde{c}_f''\ge 2$ without loss of generality, because otherwise we can set $\max(2,\widetilde{c}_f'')$ to be the new constant $\widetilde{c}_f''$.
\end{proof}

\subsection{Relative boundedness: Proof of Proposition \ref{prop:rela_bound}}\label{sec:rela_bound}
\begin{proof}[Proof of Proposition \ref{prop:rela_bound}]
Like operators in \eqref{eq:O=OSs}, any state \begin{equation}\label{eq:psi=syndrome}\ket{\psi}=\sum_{\bs}\widetilde{P}_{\bs}\ket{\psi}=:\sum_\bs \ket{\psi_\bs},
\end{equation} 
can be also decomposed into different syndrome subspaces with orthogonal projectors $\{\widetilde{P}_{\bs}\}$ for each syndrome $\bs$ defined by \begin{equation}
    \widetilde{P}_\bs := \prod_{\mathcal{Q}\in \bs}\frac{I-\mathcal{Q}}{2} \prod_{\mathcal{Q}\notin \bs}\frac{I+\mathcal{Q}}{2}.
\end{equation}
In other words, $\widetilde{P}_\bs$ projects onto the simultaneous eigensubspace of the checks with eigenvalues given in $\bs$. Note that some $\widetilde{P}_\bs=0$ if the checks have redundancies.

$D_{S,\bs}$ maps a state of syndrome $\bs'$ to syndrome $\bs+\bs'$. Furthermore, similar to the argument around \eqref{eq:PDA=0}, we have \begin{equation}\label{eq:PDP=0}
    \exists c\in \mathbb{R}: \widetilde{P}_{\bs+\bs'} (D_{S,\bs}-c_{S,\bs}I) \widetilde{P}_{\bs'}=0, \quad \text{unless } |S|\ge \widetilde{f}(r), \where r=\min_{\mathcal{Q}_C\in \bs'}\mathsf{d}(C,S).
\end{equation}
This holds because if $S$ is far from syndrome $\bs'$ comparing to its size $|S|<\widetilde{f}(r)$, $\widetilde{P}_{\bs'}\propto P_{\bar{S}}$ where $\bar{S}$ is the locally indistinguishable region in Lemma \ref{lem:S<tfr}, so that \begin{equation}\label{eq:DP=PDP}
D_{S,\bs}P_{\bar{S}}=D_{S,\bs}P_{\bar{S}}^2=P_{\bar{S}}D_{S,\bs}P_{\bar{S}}=c_{S,\bs}P_{\bar{S}},
\end{equation}
for a number $c_{S,\bs}$, which yields \eqref{eq:PDP=0}. Here the second equality in \eqref{eq:DP=PDP} holds because $D_{S,\bs}$ is block diagonal so commutes with $P_S$ and furthermore $P_{\bar{S}}$. Note that $c_{S,\bs}\neq 0$ only for $\bs=\emptyset$.

Define $c_D=\sum_{S,\bs} c_{S,\bs}$ to be the total offset of $D$, and let $\widetilde{D}=D-c_DI$.

Similar to the summation arrangement in \eqref{eq:sumS'DA<}, the condition for $S$ in \eqref{eq:PDP=0} implies the existence of site $j\in S$  that is of distance $r$ to the syndrome $\bs'$.
\eqref{eq:PDP=0} then leads to \begin{align}\label{eq:PDP<}
    \norm{\widetilde{P}_{\bs+\bs'} \widetilde{D} \widetilde{P}_{\bs}} &\le \sum_{\mathcal{Q}_C\in \bs}\sum_{i\in C}\sum_{r\ge 0: \widetilde{f}(r)<d_{\rm s}} \sum_{j:\mathsf{d}(i,j)=r} \sum_{S\ni j: |S|\ge \widetilde{f}(r)}\norm{D_{S,\bs'}-c_{S,\bs'}I} \nonumber\\
    &\le 2\sum_{\mathcal{Q}_C\in \bs}\sum_{i\in C}\sum_{r\ge 0: \widetilde{f}(r)<d_{\rm s}} \ee^{-\kappa \widetilde{f}(r)} \sum_{j:\mathsf{d}(i,j)=r} \sum_{S\ni j: |S|\ge \widetilde{f}(r)}\norm{D_{S,\bs'}} \ee^{\kappa|S|}
\end{align}
where we have used $\norm{\widetilde{P}_{\bs+\bs'} \lr{D_{S,\bs'}-c_{S,\bs'}I}\widetilde{P}_{\bs}}\le \norm{D_{S,\bs'}-c_{S,\bs'}I}$ in the first line, and $|c_{S,\bs'}|=\norm{P_{\bar{S}}D_{S,\bs'}P_{\bar{S}}}\le \norm{D_{S,\bs'}}$ in the second.
As a result, 
\begin{align}\label{eq:sumPDP} 
    \sum_{\bs'}\norm{\widetilde{P}_{\bs+\bs'} \widetilde{D} \widetilde{P}_{\bs}} &\le 2\sum_{\mathcal{Q}_C\in \bs}\sum_{i\in C}\sum_{r\ge 0: \widetilde{f}(r)<d_{\rm s}} \ee^{-\kappa \widetilde{f}(r)} \sum_{j:\mathsf{d}(i,j)=r} \sum_{S\ni j: |S|\ge \widetilde{f}(r)}\sum_{\bs'}\norm{D_{S,\bs'}} \ee^{\kappa|S|} \nonumber\\
    &\le 2 |\bs| \Delta \norm{D}_\kappa \sum_{r\ge 0: \widetilde{f}(r)<d_{\rm s}}\gamma(r) \ee^{-\kappa \widetilde{f}(r)} \le 2 \Delta|\bs|  \norm{D}_\kappa c_f'' \ee^{c_f'\kappa^{-\alpha}} \nonumber\\
    &\le 2 \Delta\norm{D}_\kappa c_f'' \ee^{c_f'(\kappa/2)^{-\alpha}} |\bs|  =:  \mathbbm{d}'|\bs|  ,
\end{align}
where we have used \eqref{eq:sum_r<} and defined $\mathbbm{d}'$ by the final expression. Similarly,
\begin{align}
    \sum_{\bs'}|\bs'|\norm{\widetilde{P}_{\bs+\bs'} \widetilde{D} \widetilde{P}_{\bs}} &\le 2\sum_{\mathcal{Q}_C\in \bs}\sum_{i\in C}\sum_{r\ge 0: \widetilde{f}(r)<d_{\rm s}} \ee^{-\frac{\kappa}{2} \widetilde{f}(r)} \sum_{j:\mathsf{d}(i,j)=r} \sum_{S\ni j: |S|\ge \widetilde{f}(r)}\sum_{\bs'}\norm{D_{S,\bs'}} \ee^{\frac{\kappa}{2}|S|}|\bs'| \nonumber\\
    &\le 2\sum_{\mathcal{Q}_C\in \bs}\sum_{i\in C}\sum_{r\ge 0: \widetilde{f}(r)<d_{\rm s}} \ee^{-\frac{\kappa}{2} \widetilde{f}(r)} \sum_{j:\mathsf{d}(i,j)=r} \sum_{S\ni j: |S|\ge \widetilde{f}(r)}\sum_{\bs'}\norm{D_{S,\bs'}} \ee^{\frac{\kappa}{2}|S|}\lr{\frac{\Delta}{\kappa}\ee^{\frac{\kappa}{2} |S|} } \nonumber\\
    &\le \frac{\Delta}{\kappa} \mathbbm{d}'|\bs|  .\label{eq:sumPDPs}
\end{align}
Here the last line comes from comparing to \eqref{eq:sumPDP}. In the second line of \eqref{eq:sumPDPs}, we have used $2x\le \ee^x$ for $x=\frac{\kappa}{2\Delta}|\bs'|\le \kappa |S|/2$, because an operator of support $S$ at most overlaps with $\Delta |S|$ checks that could belong to the syndrome $\bs'$, according to the LDPC condition. Note that this is the only place where we need each site to participate in bounded number of checks; all other places using the LDPC condition (bounding $\sum_{i\in C}\le \Delta$ here and in \eqref{eq:sumS'DA<}) invoke the other condition that each check acts on bounded number of sites. 

For any state $\psi$ expanded by \eqref{eq:psi=syndrome}, we have (here $\norm{\phi}:=\sqrt{\alr{\phi|\phi}}$)
\begin{align}\label{eq:D2<H02}
    \norm{\widetilde{D}\ket{\psi}}^2 &= \sum_{\bs_1,\bs_2,\bs_3} \bra{\psi_{\bs_3}} \widetilde{D} \widetilde{P}_{\bs_2} \widetilde{D}\ket{\psi_{\bs_1}} \le \sum_{\bs_1,\bs_2,\bs_3} \norm{\psi_{\bs_3}} \norm{\widetilde{P}_{\bs_3} \widetilde{D} \widetilde{P}_{\bs_2}} \norm{\widetilde{P}_{\bs_2} \widetilde{D} \widetilde{P}_{\bs_1}} \norm{\psi_{\bs_1}} \nonumber\\
    &\le \sum_{\bs_1,\bs_2,\bs_3}\frac{1}{2}\lr{\norm{\psi_{\bs_1}}^2+\norm{\psi_{\bs_3}}^2} \norm{\widetilde{P}_{\bs_3} \widetilde{D} \widetilde{P}_{\bs_2}} \norm{\widetilde{P}_{\bs_2} \widetilde{D} \widetilde{P}_{\bs_1}}  = \sum_{\bs_1,\bs_2}\norm{\psi_{\bs_1}}^2  \norm{\widetilde{P}_{\bs_2} \widetilde{D} \widetilde{P}_{\bs_1}} \sum_{\bs_3} \norm{\widetilde{P}_{\bs_3} \widetilde{D} \widetilde{P}_{\bs_2}}  \nonumber\\
    &\le \sum_{\bs_1,\bs_2}\norm{\psi_{\bs_1}}^2  \norm{\widetilde{P}_{\bs_2} \widetilde{D} \widetilde{P}_{\bs_1}} \lr{\mathbbm{d}' |\bs_2|} \le \mathbbm{d}' \sum_{\bs_1}\norm{\psi_{\bs_1}}^2 \sum_{\bs_2'}\lr{|\bs_1| + |\bs_2'| } \norm{\widetilde{P}_{\bs_1+\bs_2'} \widetilde{D} \widetilde{P}_{\bs_1}}  \nonumber\\
    &\le (\mathbbm{d}')^2 \sum_{\bs_1}\norm{\psi_{\bs_1}}^2 |\bs_1|\lr{|\bs_1|+\frac{\Delta}{\kappa}} \le (\mathbbm{d}')^2 \sum_{\bs_1}\norm{\psi_{\bs_1}}^2 |\bs_1|^2\lr{1+\frac{\Delta}{\kappa}} \nonumber\\
    &=\lr{1+\frac{\Delta}{\kappa}}(\mathbbm{d}')^2 \sum_{\bs_1} |\bs_1|^2 \bra{\psi}\widetilde{P}_{\bs_1}\ket{\psi}\le \lr{1+\frac{\Delta}{ \kappa}}(\mathbbm{d}')^2 \norm{H_0 \ket{\psi}}^2.
\end{align}
Here the second line comes from Cauchy-Schwarz inequality and relabeling the variables for the term $\propto \norm{\psi_{\bs_3}}^2$. To get the third line, we observe that
the sum over $\bs_3$ is equivalent to a sum over syndrome $\bs_3'$: $ \sum_{\bs_3'} \norm{\widetilde{P}_{\bs_2+\bs_3'} \widetilde{D} \widetilde{P}_{\bs_2}}$, so we can plug in \eqref{eq:sumPDP}. We have defined $\bs_2'$ similarly and used $|\bs_2|=|\bs_1+\bs_2'|\le |\bs_1|+|\bs_2'|$. We have again used \eqref{eq:sumPDP} and \eqref{eq:sumPDPs} in the fourth line of \eqref{eq:D2<H02}. The last line comes from $H_0^2 \ge \sum_\bs |\bs|^2 \widetilde{P}_\bs$ because flipping $|\bs|$ checks costs at least energy $|\bs|$ due to \eqref{eq:minlambdacheck}. \eqref{eq:D2<H02} finishes the proof of Proposition \ref{prop:rela_bound} using $\mathbbm{d}'$ defined in \eqref{eq:sumPDP}, and $\sqrt{1+x}\le 1+x/2$ for any $x\ge 0$.
\end{proof}

\section{Stability of classical codes under symmetric perturbations}\label{app:classical}
Here we present the proofs of our generalizations of Theorem \ref{thm:qtm} to classical codes with symmetric perturbations.

\begin{repcor}{cor:cla}[Formal version]
    Theorem \ref{thm:qtm} also holds for $H_0$ being a classical LDPC code, as long as each term $V_{S,\bs}$ in the perturbation $V$ is symmetric under conjugation by any logical operator (which are Pauli-$X$s by Definition \ref{def:stabilizer_code}).
\end{repcor}

\begin{proof}
    The proof closely follows the argument in the discussion section of \cite{bravyi2011short}.
    Here the only difference comparing to the proof of Theorem \ref{thm:qtm} is that \eqref{eq:local_indisting} (where $\bar{S}$ is defined in the same way as quantum codes from check soundness) does not hold for all operators $\OO$, because a single $Z_i$ could differentiate the ground states. An alternative viewpoint is that $Z_i$ is a logical $Z$ operator when viewing the classical code as a quantum CSS code with only $Z$ checks.

    Nevertheless, \eqref{eq:local_indisting} holds for any $\OO$ that is symmetric under conjugation by $X$ logicals, because this precisely rules out $Z$ logicals that must anti-commute with some $X$ logical. If the local perturbation $V_{S,\bs}$ is symmetric as we assumed, the operators in the proof of Theorem \ref{thm:qtm} are all symmetric because they are ``generated'' from the perturbation, along with $H_0$ which is also symmetric. More precisely, the initial perturbation $V_1$ first generates (local terms of) $A_1$ by \eqref{eq:A=PVQH}, which is locally symmetric because the $P,Q,H$ factors in \eqref{eq:A=PVQH} all commute with $X$ logicals (since each check in $H_0$ commutes with $X$ logicals). The same holds for $D_2$ in \eqref{eq:D=D+PV}. $V_2$ generated by \eqref{eq:Vk+1} and \eqref{eq:Vm+1=} is also locally symmetric because all components in \eqref{eq:Vk+1} are symmetric. This iterates to arbitrary order $m$, so that we only encounter symmetric $\OO$ when applying \eqref{eq:local_indisting}. The proof of Theorem \ref{thm:qtm} thus generalizes here.
\end{proof}

\begin{repcor}{cor:Ising}[Formal version]
    The conclusions of Theorem \ref{thm:qtm} hold for $H_0$ being the Ising model \eqref{eq:IsingH0} on graph $G$, where the parameter $d_{\rm s}$ may be taken to be $d_{\mathrm{s}}=n$, if the symmetric perturbation $V=\sum_{S,\bs }V_{S,\bs}$ where each $S$ is connected in $G$. 
\end{repcor}

\begin{proof}
We first prove that \eqref{eq:local_indisting} holds for any connected $S$ and symmetric $\OO$ with $\bar{S}=\widetilde{S}$ being $S$ and its neighborhood. In other words, \eqref{eq:local_indisting0} holds: Since $S$ is connected, $P_{\widetilde{S}}=\ket{\mathbf{0}}_{\widetilde{S}}\bra{\mathbf{0}} + \ket{\mathbf{1}}_{\widetilde{S}}\bra{\mathbf{1}}$ projects onto the all-zero $\ket{\mathbf{0}}$ and all-one $\ket{\mathbf{1}}$ states in $\widetilde{S}$. If $\OO$ is symmetric under the $X$ logical $\prod_{i\in \Lambda} X_i$, it is also symmetric under its local version $\mathbf{X}:=\prod_{i\in \widetilde{S}} X_i$, so that \begin{equation}
    \bra{\mathbf{0}} \OO \ket{\mathbf{0}}_{\widetilde{S}} = \bra{\mathbf{0}}\mathbf{X} \OO \mathbf{X}\ket{\mathbf{0}}_{\widetilde{S}}=\bra{\mathbf{1}} \OO \ket{\mathbf{1}}_{\widetilde{S}},
\end{equation}
and $\bra{\mathbf{1}} \OO \ket{\mathbf{0}}_{\widetilde{S}}=0$ because $\OO$ does not act on $\widetilde{S}\setminus S$. As a result, $\OO$ acts as an identity in the local ground states and \eqref{eq:local_indisting0} holds.

Observe that by restricting the initial perturbation to contain only connected $S$, all generated operators $A_m,D_m,V_m$ in the proof of Theorem \ref{thm:qtm} also only contain connected supports $S$ for their local terms, which can be checked in a similar way as the proof of Corollary \ref{cor:cla}. In particular, the connectedness pertains when invoking Lemma \ref{lem:clustexp} to generate operators like $(\ee^{\cA_m}-1)V_m$; see \cite{abanin2017rigorous,our_metastable} for details, which indeed work with connected sets throughout. 

Corollary \ref{cor:Ising} then follows from the above two observations: As we only need to consider connected sets of qubits, we can use \eqref{eq:local_indisting0} directly that avoids invoking check soundness. More precisely, check soundness was only used to obtain \eqref{eq:sum_r<}, but \eqref{eq:sum_r<} now trivially holds because $\bar{S}$ is only $S$ with its neighborhood so that we can set $\widetilde{f}(r)=\infty$ for $r>1$.
\end{proof}

\end{appendix}

\bibliography{thebib}

\end{document}